\documentclass[journal]{IEEEtran}
\usepackage{setspace}
\usepackage{comment}
\usepackage{url}

\usepackage{graphicx}
\usepackage{amssymb}
\usepackage{amsmath,amsfonts,amssymb,euscript,epsfig,psfrag,
enumerate,float,afterpage, subfigure}%

\newtheorem{thm}{Theorem}

\newtheorem{lem}{Lemma}
\newcommand{\expect}[1]{\mathbb{E}\left\{#1\right\}}
\newcommand{\defequiv}{\mbox{\raisebox{-.3ex}{$\overset{\vartriangle}{=}$}}}

\usepackage{cite}

\mark{PROC. SIGMETRICS 2011 (TO APPEAR)}{}

\begin{document}

\title{Optimal Power Cost Management Using Stored Energy in Data Centers}

\author{
{Rahul Urgaonkar}, {Bhuvan Urgaonkar}$\dag$, {Michael J. Neely}, and {Anand Sivasubramaniam}$\dag$\\
\begin{tabular}{cccc}
Dept. of EE-Systems,&$\dag$ Dept. of CSE,\\
University of Southern California& The Pennsylvania State University,\\
\{urgaonka,mjneely\}@usc.edu& \{bhuvan,anand\}@cse.psu.edu
\end{tabular}
}


\maketitle
\begin{abstract}
Since the electricity bill of a data center constitutes a significant portion of
its overall operational costs, reducing this has become important. We investigate
cost reduction opportunities that arise by the use of uninterrupted power supply
(UPS) units as energy storage devices. This represents a deviation from the usual
use of these devices as mere transitional fail-over mechanisms between utility
and captive sources such as diesel generators.   
We consider the problem of opportunistically using these devices 
to reduce the time average electric utility bill in a data center.
Using the technique of Lyapunov optimization, we develop an online control
algorithm that can optimally exploit these devices to minimize the time average cost. This algorithm
operates without any knowledge of the statistics of the workload or electricity cost processes, making it
attractive in the presence of workload and pricing uncertainties. 
An interesting feature of our algorithm is that its deviation from optimality reduces as the storage capacity is increased.
Our work opens up a new area in data center power management.

\end{abstract}


       

\section{Introduction}
\label{section:intro}


Data centers spend a significant portion of their overall operational
costs towards their electricity bills. As an example, one recent 
case study suggests that a large 15MW data center (on the more 
energy-efficient end) might spend about \$1M on its monthly electricity bill. 
In general, a data center spends between 30-50\% of its operational expenses
towards power. A large body of research addresses these expenses by 
reducing the energy consumption of these data centers. This includes 
designing/employing hardware with better power/performance trade-offs~\cite{Soyeon07,Zhu04,Gurumurthi03}, 
software techniques for power-aware scheduling~\cite{Lebeck00}, workload migration, 
resource consolidation~\cite{Muse}, among others. Power prices exhibit 
variations along time, space (geography), and even across utility providers.
As an example, consider Fig. \ref{fig:one_week} that shows the average hourly spot market prices
for the Los Angeles Zone LA1 obtained from CAISO~\cite{caiso}.
These correspond to the week of 01/01/2005-01/07/2005
and denote the average price of 1 MW-Hour of electricity.
Consequently, minimization of energy consumption need not coincide with that of the electricity bill.

\begin{figure}
\centering
\includegraphics[height=3.5cm, width=8.5cm,angle=0]{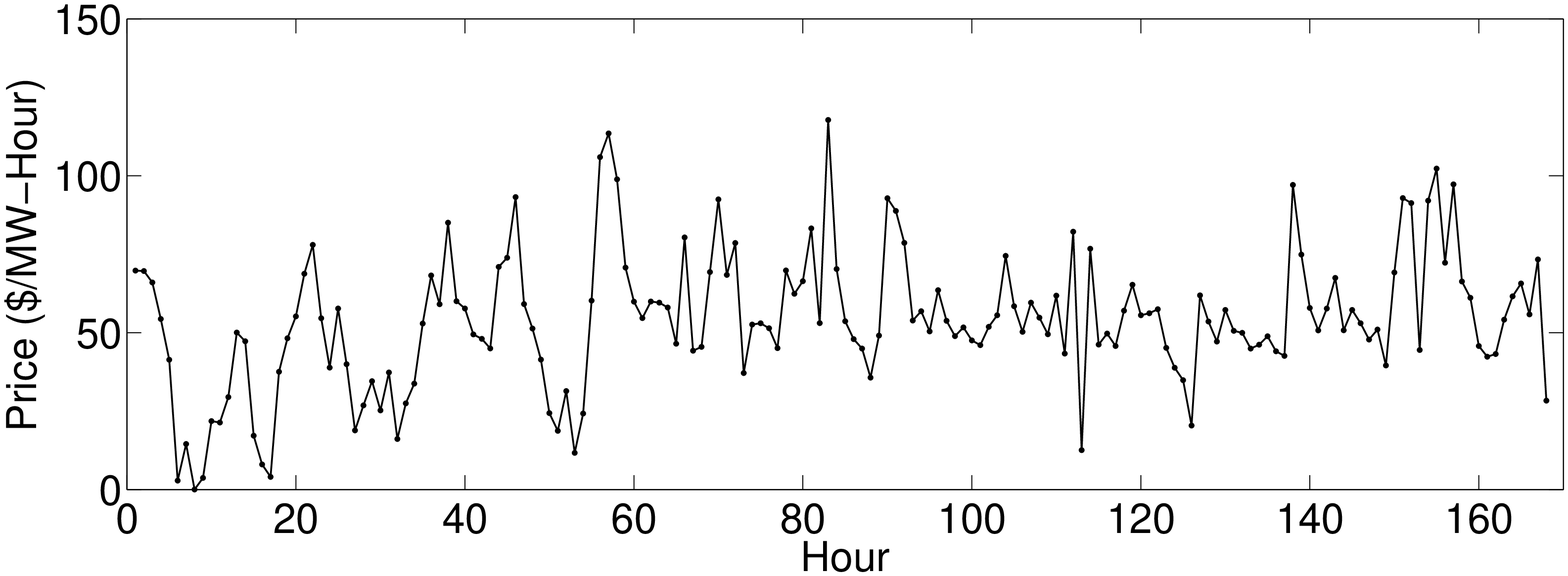}
\caption{Average hourly spot market price during the week of 01/01/2005 - 01/07/2005 for LA1 Zone \cite{caiso}.}
\label{fig:one_week}
\end{figure}

Given the diversity within power price and availability, attention has
recently turned towards {\em demand response} (DR) within 
data centers. DR within a data center (or a set of related data centers) 
attempts to optimize the electricity bill by adapting its needs to the
temporal, spatial, and cross-utility diversity exhibited by power price. 
The key idea behind these techniques is to preferentially 
shift power draw (i) to times and places or (ii) from utilities offering 
cheaper prices. Typically some constraints in the form of performance 
requirements for the workload (e.g., response times offered to the clients of 
a Web-based application) limit the cost reduction benefits that can result 
from such DR. Whereas existing DR techniques have relied on various forms
of workload scheduling/shifting, a complementary knob to facilitate such 
movement of power needs is offered by {\em energy storage devices}, typically
uninterrupted power supply (UPS) units, residing in data centers.

A data center
deploys captive power sources, typically diesel generators (DG), that it uses
for keeping itself powered up when the utility experiences an outage. The
UPS units serve as a bridging mechanism to facilitate this transition from 
utility to DG: upon a utility failure, the data center is kept powered by
the UPS unit using energy stored within its batteries, before the DG
can start up and provide power. Whereas this transition takes only 10-20 seconds, 
UPS units have enough battery capacity to keep the entire data center powered
at its maximum power needs for anywhere between 5-30 minutes. 
Tapping into the energy reserves
of the UPS unit can allow a data center to improve its electricity bill. 
Intuitively, the  data center would store energy within the UPS unit when 
prices are low and use this to augment the draw from the utility when prices 
are high.

In this paper, we consider the problem of developing an 
online control policy to exploit the UPS unit along with the presence of
delay-tolerance within the workload to optimize the data center's
electricity bill. This is a challenging problem because 
data centers experience time-varying workloads and power 
prices with possibly unknown statistics. Even when statistics can be 
approximated (say by learning using past observations), traditional approaches 
to construct optimal control policies involve the use of Markov Decision
Theory and Dynamic Programming \cite{DP}. It is well known that these techniques
suffer from the ``curse of dimensionality'' where the complexity of
computing the optimal strategy grows with the system size. Furthermore, 
such solutions result in hard-to-implement systems, where significant
re-computation might be needed when statistics change. 

In this work, we make use of a different approach that can overcome the challenges 
associated with dynamic programming. This approach is based on
the recently developed technique of Lyapunov optimization \cite{neely-NOW}\cite{neely_new} 
that enables the design of online control algorithms for such time-varying systems. These
algorithms operate \emph{without requiring any knowledge of the system statistics} 
and are easy to implement. We design such an algorithm for optimally exploiting the UPS unit and
delay-tolerance of workloads to minimize the time average cost. 
We show 
that our algorithm  can get within $O(1/V)$ of the optimal solution where the maximum value of $V$ is limited by battery capacity. 
We note that, for the same parameters, a dynamic programming based approach (if it can be solved) will yield a better result than our
algorithm. However, this gap reduces as the battery capacity is increased.
Our algorithm is thus most useful when such scaling is practical.





\section{Related Work}
\label{sec:back}

One recent body of work proposes online algorithms for using UPS 
units for cost reduction via shaving workload ``peaks'' 
that correspond to 
higher energy prices~\cite{CUNYPeakShaving08,CUNYPeakShaving08_2}. This work is highly complementary to ours 
in that it offers a worst-case competitive ratio analysis while our
approach looks at the average case performance.  
Whereas a variety of work has looked at workload shifting for power cost
reduction~\cite{Zhu04} or other reasons such as performance and 
availability~\cite{Muse},  our work differs both due to its usage
of energy storage as well as  
the cost optimality guarantees offered by our technique. Some research
has considered consumers with  access to multiple utility providers, 
each with a different carbon profile, power price and availability and 
looked at optimizing cost subject to performance and/or carbon emissions
constraints~\cite{bianchini09_2}. 
Another line of work has looked at cost reduction
opportunities offered by geographical variations within utility prices for
data centers where portions of workloads could be serviced from one of
several locations~\cite{bianchini09_2,Qureshi09}. Finally, \cite{marios} considers
the use of rechargeable batteries for maximizing system utility in a wireless network. 
 While all of this research
is highly complementary to our work, there are three key differences: 
(i) our investigation of energy storage as an enabler of cost reduction, 
(ii) our use of the technique of Lyapunov optimization which allows us to offer a provably cost optimal solution,
and (iii) combining energy storage with delay-tolerance within workloads.





\section{Basic Model}
\label{section:basic_model}

We consider a time-slotted model. In the basic model, we assume that in every slot, the total power demand generated by the data center 
in that slot must be met in the current slot itself (using a combination of power drawn from the utility and the battery). Thus, any buffering of the 
workload generated by the data center is not  allowed. We will relax this constraint later in Sec. \ref{section:ext_model} when we allow buffering of
some of the workload while providing worst case delay guarantees.
In the following, we use the terms UPS and battery interchangeably. 

\begin{figure}
\centering
\includegraphics[height=2.5cm, width=7.5cm,angle=0]{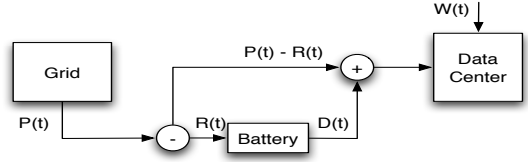}
\caption{Block diagram for the basic model.}
\label{fig:one}
\end{figure}

\subsection{Workload Model}
\label{section:workload_model}

Let $W(t)$ be total workload (in units of power) generated in slot $t$. 
Let $P(t)$ be the total power drawn from the grid in slot $t$ out of which $R(t)$ is used to recharge the battery. 
Also, let $D(t)$ be the total power discharged from the battery in slot $t$. Then in the basic model, the
following constraint must be satisfied in every slot (Fig. \ref{fig:one}):
\begin{align}
W(t) = P(t) - R(t) + D(t)
\label{eq:w_t}
\end{align}
Every slot, a control algorithm observes $W(t)$ and makes decisions about how much power to draw from the grid in that slot, i.e., $P(t)$,
and how much to recharge and discharge the battery, i.e., $R(t)$ and $D(t)$. Note that by (\ref{eq:w_t}), having chosen $P(t)$ and $R(t)$
completely determines $D(t)$.

\emph{Assumptions on the statistics of $W(t)$}: The workload process $W(t)$ is assumed to vary randomly taking values from 
a set $\mathcal{W}$ of non-negative values and is not influenced by past control decisions. 
The set $\mathcal{W}$ is assumed to be finite, with potentially arbitrarily large size. 
The underlying probability distribution or statistical characterization of $W(t)$
is not necessarily known. 
We only assume that its maximum value is finite, i.e., $W(t) \leq W_{max}$ for all $t$.
Note that unlike existing work in this domain, we do not make assumptions such as Poisson arrivals or exponential
service times.


For simplicity, in the basic model we assume that $W(t)$ evolves according to an i.i.d. process noting that
the algorithm developed for this case can be applied without any modifications to non i.i.d. scenarios as well.
The analysis and performance guarantees for the non i.i.d. case can be obtained
using the delayed Lyapunov drift and $T$ slot drift techniques developed in \cite{neely-NOW}\cite{neely_new}. 


\subsection{Battery Model}
\label{section:battery_model}


Ideally, we would like to incorporate the following idiosyncrasies of
battery operation into our model. First, batteries become unreliable as they
are charged/discharged, with higher depth-of-discharge (DoD) - percentage of 
maximum charge removed during a discharge cycle - causing faster degradation 
in their reliability. This dependence between the useful lifetime of a 
battery and how it is discharged/charged is expressed via battery lifetime 
charts~\cite{BatteryBook}. For example, with lead-acid batteries that are commonly 
used in UPS units, 20\% DoD yields $1400$ cycles ~\cite{UPS_DepthOfDischarge}.
Second, batteries have conversion loss whereby a portion of the energy 
stored in them is lost when discharging them (e.g., about 10-15\% for
lead-acid batteries). Furthermore, certain regions of battery operation
(high rate of discharge) are more inefficient than others.  
Finally, the storage itself maybe ``leaky'', so that
the stored energy decreases over time, even in the absence of any discharging.

For simplicity, in the basic model we will assume that there is no power loss either 
in recharging or discharging the batteries, noting that this can be easily generalized to
the case where a fraction of $R(t), D(t)$ is lost. 
We will also assume that the batteries  are not leaky, so that the stored energy level decreases only when they are 
discharged. This is a reasonable assumption when the time scale over which the 
loss takes place is much larger than that of interest to us. To model the effect of 
repeated recharging and discharging on the battery's lifetime, we assume that with 
each recharge and discharge operation, a fixed cost (in dollars) of $C_{rc}$ and 
$C_{dc}$ respectively is incurred. 
The choice of these  parameters would affect the trade-off between the cost of the battery itself and the cost 
reduction benefits it offers. For example, suppose a new battery costs $B$ dollars and 
it can sustain $N$ discharge/charge cycles (ignoring DoD for now). Then 
setting $C_{rc} = C_{dc} = B/N$ would amount to expecting the battery to ``pay for itself'' by augmenting the utility $N$ times over its lifetime.


In any slot, we assume that one can either recharge or discharge the battery or do neither, but not both. This means 
that for all $t$, we have:
\begin{align}
R(t) > 0 \implies D(t) = 0, \; D(t) > 0 \implies R(t) = 0
\label{eq:rt_dt}
\end{align}
Let $Y(t)$ denote the battery energy level
in slot $t$. Then, the dynamics of $Y(t)$ can be expressed as:
\begin{align}
Y(t+1) = Y(t) - D(t) + R(t)
\label{eq:y_t}
\end{align}
The battery is assumed to have a finite capacity $Y_{max}$ so that $Y(t) \leq Y_{max}$ for all $t$. 
Further, for the purpose of reliability, it may be required to ensure that a minimum energy level $Y_{min} \geq 0$ is maintained at 
all times. For example, this could represent the amount of energy required to support the data center operations until a secondary 
power source (such as DG) is activated in the event of a grid outage.
Recall that the UPS unit is integral to the availability of power supply to the data center upon utility outage. 
Indiscriminate discharging of UPS can leave the data center in situations where it is unable to safely fail-over 
to DG upon a utility outage. Therefore, discharging the UPS must be done carefully so that it still possesses enough 
charge so reliably carry out its role as a transition device between utility and DG.
Thus, the following condition must be met in every slot under any feasible control algorithm:
\begin{align}
Y_{min} \leq Y(t) \leq Y_{max}
\label{eq:y_t_bound}
\end{align}

The effectiveness of the online control algorithm we present in Sec. \ref{section:optimal_algorithm} will depend on
the magnitude of the difference $Y_{max} - Y_{min}$. In most practical scenarios of interest, this value is expected to be at least
moderately large: recent work suggests that storing energy $Y_{min}$ to last about a 
minute is sufficient to offer reliable data center operation~\cite{hp_avail_10}, 
while $Y_{max}$ can vary between 5-20 minutes (or even higher) due to reasons such as UPS units being available only in certain sizes 
and the need to keep room for future IT growth. Furthermore, the UPS units are sized based on the \emph{maximum provisioned} capacity 
of the data center, which is itself often substantially (up to twice~\cite{GoogleBook}) higher than the maximum actual power demand. 


The initial charge level in the battery is given by $Y_{init}$ and satisfies $Y_{min} \leq Y_{init} \leq Y_{max}$.
Finally, we assume that the maximum amounts by which we can recharge or discharge the battery in any slot are bounded. Thus, we have 
$\forall t$:
\begin{align}
0 \leq R(t) \leq R_{max}, \; 0 \leq D(t) \leq D_{max}
\label{eq:rt_dt_bound}
\end{align}
We will assume that $Y_{max} - Y_{min} > R_{max} + D_{max}$ while noting that
in practice, $Y_{max} - Y_{min}$ is much larger than $R_{max} + D_{max}$. 
Note that any feasible control decision on $R(t), D(t)$ must ensure that both of the constraints (\ref{eq:y_t_bound}) and (\ref{eq:rt_dt_bound}) are satisfied. 
This is equivalent to the following:
\begin{align}
&0 \leq R(t) \leq \min[R_{max}, Y_{max} - Y(t)] \label{eq:rt_bound} \\
&0 \leq D(t) \leq \min[D_{max}, Y(t) - Y_{min}] \label{eq:dt_bound}
\end{align}

\subsection{Cost Model}
\label{section:cost_model}

The cost per unit of power drawn from the grid in slot $t$ is denoted by $C(t)$. In general, it can depend on both $P(t)$, 
the total amount of power drawn in slot $t$, and an auxiliary state variable $S(t)$, that captures parameters such as time of day, 
identity of the utility provider, etc. For example, the per unit cost may be higher during business hours, etc. 
Similarly, for any fixed $S(t)$, it may be the case that $C(t)$ increases with $P(t)$ so that per unit cost of electricity 
increases as more power is drawn. This may be because the utility provider wants to discourage heavier loading on the grid. 
Thus, we assume that $C(t)$ is a function of both $S(t)$ and $P(t)$ and we denote this as:
\begin{align}
C(t) = \hat{C}(S(t), P(t))
\label{eq:c_hat}
\end{align}
For notational convenience, we will use $C(t)$ to denote the per unit cost in the rest of the paper noting
that the dependence of $C(t)$ on $S(t)$ and $P(t)$ is implicit.

The auxiliary state process $S(t)$ is assumed to evolve independently of the decisions taken by any control policy. 
For simplicity, we assume that every slot it takes values from a finite but arbitrarily large set $\mathcal{S}$ 
in an i.i.d. fashion according to a potentially unknown distribution.
This can again be generalized to non i.i.d. Markov modulated scenarios using the techniques developed in \cite{neely-NOW}\cite{neely_new}.
For each $S(t)$, the unit cost is assumed to be a non-decreasing function of $P(t)$. 
Note that it is not necessarily convex or strictly monotonic or continuous.
This is quite general and can be used to model a variety of scenarios.
A special case is when $C(t)$ is only a function of $S(t)$. The optimal control action for this case has
a particularly simple form and we will highlight this in Sec. \ref{section:case1}.
The unit cost is assumed to be non-negative and finite for all $S(t), P(t)$.

We assume that the maximum amount of power that can be drawn from the grid in any slot is upper bounded by $P_{peak}$. Thus, we have for all $t$:
\begin{align}
0 \leq P(t) \leq P_{peak}
\label{eq:p_peak}
\end{align}
Note that if we consider the original scenario where batteries are not used, then $P_{peak}$ must be such that all workload can be satisfied. Therefore,
$P_{peak} \geq W_{max}$. 

Finally, let $C_{max}$ and $C_{min}$ denote the maximum and minimum per unit cost respectively over all $S(t), P(t)$.
Also let $\chi_{min} > 0$ be a constant such that for any $P_1, P_2 \in [0, P_{peak}]$ where $P_1 \leq P_2$,
the following holds for all $\chi \geq \chi_{min}$:
\begin{align}
P_1(-\chi + C(P_1, S)) \geq P_2(-\chi + C(P_2, S)) \qquad \forall S \in \mathcal{S}
\label{eq:chi_defn}
\end{align}
For example, when $C(t)$ does not depend on $P(t)$, then $\chi_{min} = C_{max}$ satisfies (\ref{eq:chi_defn}). 
This follows by noting that 
$(-C_{max} + C(t)) \leq 0$ for all $t$. Similarly, suppose $C(t)$ does not depend on $S(t)$, but is
continuous, convex, and increasing in $P(t)$. Then, it can be shown that
$\chi_{min} = C(P_{peak}) + P_{peak} C'(P_{peak})$ satisfies (\ref{eq:chi_defn}) where $C'(P_{peak})$ denotes the
derivative of $C(t)$ evaluated at $P_{peak}$.
In the following, we assume that such a finite $\chi_{min}$ exists for the given cost model.
We further assume that $\chi_{min} > C_{min}$. The case of $\chi_{min} = C_{min}$ corresponds to the
degenerate case where the unit cost is fixed for all times and we do not consider it.


\emph{What is known in each slot?}: We assume that the value of $S(t)$ and the form of the function $C(P(t), S(t))$ for that slot is known.
For example, this may be obtained beforehand using pre-advertised prices by the utility provider. 
We assume that given an $S(t)=s$, $C(t)$ is a deterministic function of $P(t)$ and this holds for all $s$. Similarly,
the amount of incoming workload $W(t)$ is known at the beginning of each slot.






Given this model, our goal is to design a control algorithm that minimizes the time average cost while meeting all the constraints.
This is formalized in the next section.

\section{Control Objective}
\label{section:objective}

Let $P(t), R(t)$ and $D(t)$ denote the control decisions made in slot $t$ by any feasible policy under the basic model
as discussed in Sec. \ref{section:basic_model}. These must satisfy the constraints
(\ref{eq:w_t}), (\ref{eq:rt_dt}), (\ref{eq:rt_bound}), (\ref{eq:dt_bound}), and (\ref{eq:p_peak}) every slot.
We define the following indicator variables that are functions of the control decisions
regarding a recharge or discharge operation in slot $t$:
\begin{displaymath}
1_R(t) = \left\{ \begin{array}{ll}
1 & \textrm{if $R(t) > 0$}\\
0 & \textrm{else}
\end{array} \right. \;\;
1_D(t) = \left\{ \begin{array}{ll}
1 & \textrm{if $D(t) > 0$}\\
0 & \textrm{else}
\end{array} \right.
\end{displaymath}
Note that by (\ref{eq:rt_dt}), at most one of $1_R(t)$ and $1_C(t)$ can take the value $1$.
Then the total cost incurred in slot $t$ is given by
$P(t) C(t) + 1_R(t) C_{rc} + 1_D(t) C_{dc}$.
The time-average cost under this policy is given by:
\begin{align}
\lim_{t\rightarrow\infty} \frac{1}{t} \sum_{\tau = 0}^{t-1} \expect{P(\tau) C(\tau) + 1_R(\tau) C_{rc} + 1_D(\tau) C_{dc}}
\label{eq:nine}
\end{align}
where the expectation above is with respect to the potential randomness of the control policy.
Assuming for the time being that this limit exists, our goal is to design a control algorithm that minimizes this time average 
cost subject to the constraints described in the basic model. Mathematically, this can be stated as the following 
\emph{stochastic optimization problem}:
\begin{align*}
& \textrm{\bf{P1}}: \\
&\textrm{Minimize:} \;\;  \lim_{t\rightarrow\infty} \frac{1}{t} \sum_{\tau = 0}^{t-1} \expect{P(\tau) C(\tau) + 1_R(\tau) C_{rc} + 1_D(\tau) C_{dc}} \\
& \textrm{Subject to:}\;\; \textrm{Constraints} \; (\ref{eq:w_t}), (\ref{eq:rt_dt}), (\ref{eq:rt_bound}), (\ref{eq:dt_bound}), (\ref{eq:p_peak}) 
\end{align*}
The finite capacity and underflow constraints (\ref{eq:rt_bound}), (\ref{eq:dt_bound}) make this a particularly challenging problem to solve
even if the statistical descriptions of the workload and unit cost process are known. For example, the traditional 
approach based on Dynamic Programming \cite{DP} would have to compute the optimal control action for all possible combinations of
the battery charge level and the system state $(S(t), W(t))$.
Instead, we take an alternate approach based on the technique of Lyapunov optimization, taking the
\emph{finite} size queues constraint explicitly into account.


Note that a solution to the problem \textbf{P1} is a control policy that determines the sequence of feasible control decisions $P(t)$, $R(t)$, $D(t)$, 
to be used. Let $\phi_{opt}$ denote the value of the objective in problem \textbf{P1} under an optimal control policy. 
Define the time-average rate of recharge and discharge under any policy as follows:
\begin{align}
\overline{R} = \lim_{t\rightarrow\infty} \frac{1}{t} \sum_{\tau = 0}^{t-1} \expect{R(\tau)}, \; 
\overline{D} = \lim_{t\rightarrow\infty} \frac{1}{t} \sum_{\tau = 0}^{t-1} \expect{D(\tau)}
\label{eq:r_bar_d_bar} 
\end{align}
Now consider the following problem:
\begin{align}
& \textrm{\bf{P2}}: \nonumber\\
&\textrm{Minimize:} \;\;  \lim_{t\rightarrow\infty} \frac{1}{t} \sum_{\tau = 0}^{t-1} \expect{P(\tau) C(\tau) + 1_R(\tau) C_{rc} + 1_D(\tau) C_{dc}} \nonumber\\
&\textrm{Subject to:}\;\; \textrm{Constraints} \; (\ref{eq:w_t}), (\ref{eq:rt_dt}), (\ref{eq:rt_dt_bound}), (\ref{eq:p_peak}) \nonumber\\
& \qquad \qquad \;\;\;\;\; \overline{R} = \overline{D} \label{eq:eleven}
\end{align}
Let $\hat{\phi}$ denote the value of the objective in problem \textbf{P2} under an optimal control policy. 
By comparing \textbf{P1} and \textbf{P2}, it can be shown that \textbf{P2} is less constrained than \textbf{P1}. Specifically, 
any feasible solution to \textbf{P1} would also satisfy \textbf{P2}. To see this, consider any policy that satisfies (\ref{eq:rt_bound}) and (\ref{eq:dt_bound})
for all $t$. This ensures that constraints (\ref{eq:y_t_bound}) and (\ref{eq:rt_dt_bound}) are always met by this policy. Then summing equation (\ref{eq:y_t}) over 
all $\tau \in \{0, 1, 2, \ldots, t-1\}$ under this policy and taking expectation of both sides yields:
\begin{align*}
\expect{Y(t)} - Y_{init} = \sum_{\tau =0}^{t-1} \expect{R(\tau) - D(\tau)}
\end{align*}
Since $Y_{min} \leq Y(t) \leq Y_{max}$ for all $t$, dividing both sides by $t$ and taking limits as $t\rightarrow\infty$ yields $\overline{R} = \overline{D}$.
Thus, this policy satisfies constraint (\ref{eq:eleven}) of \textbf{P2}. Therefore, any feasible solution to \textbf{P1} also satisfies \textbf{P2}.
This implies that the optimal value of \textbf{P2} cannot exceed that of \textbf{P1}, so that $\hat{\phi} \leq \phi_{opt}$.

Our approach to solving \textbf{P1} will be based on this observation. We first note that it is easier to characterize the optimal solution to 
\textbf{P2}. This is because the dependence on $Y(t)$ has been removed.
Specifically, it can be shown that the optimal solution to \textbf{P2} can be achieved by
a stationary, randomized control policy that chooses control actions $P(t), D(t), R(t)$ every slot purely as a function (possibly randomized)  
of the current state $(W(t), S(t))$ and \emph{independent of} the battery charge level $Y(t)$. This fact is presented in the following lemma:
\vspace{-0.1in}
\begin{lem} (Optimal Stationary, Randomized Policy): If the workload process $W(t)$ and auxiliary process $S(t)$ are i.i.d. over slots, then there
exists a stationary, randomized policy that takes control decisions $P^{stat}(t), R^{stat}(t), D^{stat}(t)$
every slot purely as a function (possibly randomized) of the current state $(W(t), S(t))$ while satisfying the
 constraints (\ref{eq:w_t}), (\ref{eq:rt_dt}), (\ref{eq:rt_dt_bound}), (\ref{eq:p_peak}) and providing the following guarantees:
\begin{align}
&\expect{R^{stat}(t)} = \expect{D^{stat}(t)} \label{eq:stat1} \\
&\expect{P^{stat}(t) C(t) + 1_R^{stat}(t) C_{rc} + 1_D^{stat}(t) C_{dc}} = \hat{\phi} \label{eq:stat2}
\end{align}
where the expectations above are with respect to the stationary distribution of $(W(t), S(t))$ and the randomized control decisions.
\label{lem:one}
\end{lem}

\begin{proof} This result follows from the framework in \cite{neely-NOW, neely_new} and is omitted for brevity.
\end{proof}

It should be noted that while it is possible to characterize and potentially compute such a policy,
it may not be feasible for the original problem \textbf{P1} as it could
violate the constraints (\ref{eq:rt_bound}) and (\ref{eq:dt_bound}). However, the existence of such a policy can be used to construct an approximately
optimal policy that meets all the constraints of \textbf{P1} using the technique of Lyapunov optimization \cite{neely-NOW}\cite{neely_new}. 
This policy is dynamic and does not
require knowledge of the statistical description of the workload and cost processes.
We present this policy and derive its performance guarantees in the next section. 
This dynamic policy is approximately optimal where the approximation factor improves as the battery capacity increases. 
Also note that the distance from
optimality for our policy is measured in terms of $\hat{\phi}$. However, since $\hat{\phi} \leq \phi_{opt}$, in practice, 
the approximation factor is better than the analytical bounds.

\section{Optimal Control Algorithm}
\label{section:optimal_algorithm}

We now present an \emph{online} control algorithm that approximately solves \textbf{P1}. This algorithm uses a control parameter $V > 0$ that affects the distance from
optimality as shown later. This algorithm also makes use of a ``queueing'' state variable $X(t)$ to track the battery charge level and is defined as follows:
\begin{align}
X(t) = Y(t) - V\chi_{min} - D_{max} - Y_{min}
\label{eq:x_t_def}
\end{align}
Recall that $Y(t)$ denotes the actual battery charge level in slot $t$ and evolves according to
(\ref{eq:y_t}). It can be seen that $X(t)$ is simply a shifted version of $Y(t)$ and its dynamics is given by: 
\begin{align}
X(t+1) = X(t) - D(t) + R(t)
\label{eq:x_t}
\end{align}
Note that $X(t)$ can be negative. 
We will show that this definition enables our algorithm to ensure that the constraint (\ref{eq:y_t_bound}) is met.

We are now ready to state the dynamic control algorithm. 
Let $(W(t), S(t))$ and $X(t)$ denote the system state in slot $t$. Then the dynamic algorithm chooses control action $P(t)$ as the 
solution to the following optimization problem:
\begin{align}
& \textrm{\bf{P3}}: \nonumber\\
&\textrm{Minimize:} \;  X(t)P(t) + V \Big[P(t) {C}(t) + 1_R(t)C_{rc} + 1_D(t) C_{dc}\Big] \nonumber\\
&\textrm{Subject to:} \;\textrm{Constraints} \; (\ref{eq:w_t}), (\ref{eq:rt_dt}), (\ref{eq:rt_dt_bound}), (\ref{eq:p_peak}) \nonumber
\end{align}
The constraints above result in the following constraint on $P(t)$:
\begin{align}
P_{low} \leq P(t) \leq P_{high} 
\label{eq:p_t_constraint}
\end{align}
where\\
$P_{low} = \max[0, W(t) - D_{max}]$ and $P_{high} = \min[P_{peak}, W(t) + R_{max}]$. 
Let $P^*(t)$ denote the optimal solution to \textbf{P3}. Then, the dynamic algorithm chooses the recharge and discharge values as follows.
\begin{displaymath}
R^*(t) = \left\{ \begin{array}{ll}
P^*(t) - W(t) & \textrm{if $P^*(t) > W(t)$}\\
0 & \textrm{else}
\end{array} \right.
\end{displaymath}
\begin{displaymath}
D^*(t) = \left\{ \begin{array}{ll}
W(t) - P^*(t) & \textrm{if $P^*(t) < W(t)$}\\
0 & \textrm{else}
\end{array} \right.
\end{displaymath}
Note that if $P^*(t) = W(t)$, then both $R^*(t)=0$ and $D^*(t)=0$ and all demand is met using power drawn from the grid.
It can be seen from the above that the control decisions satisfy the constraints
$0 \leq R^*(t) \leq R_{max}$ and $0 \leq D^*(t) \leq D_{max}$. That the finite battery constraints 
and the constraints (\ref{eq:rt_bound}), (\ref{eq:dt_bound}) are also met will be shown in 
Sec. \ref{section:performance_analysis}.

After computing these quantities, the algorithm implements them and updates the queueing variable $X(t)$ according to (\ref{eq:x_t}). This process is
repeated every slot. Note that in solving \textbf{P3}, the control algorithm only makes use of the current system state values and does not require
knowledge of the statistics of the workload or unit cost processes. Thus, it is \emph{myopic} and \emph{greedy} in nature. 
From \textbf{P3}, it is seen that the algorithm tries to recharge the battery when $X(t)$ is negative and per unit cost is low. 
And it tries to discharge the battery when $X(t)$ is positive. 
That this is sufficient to achieve optimality will be shown in Theorem \ref{thm:one}.
The queueing variable $X(t)$ plays a crucial role as making decisions purely based on prices is not necessarily optimal.

\begin{figure}
\centering
\includegraphics[height=2.5cm, width=7.5cm,angle=0]{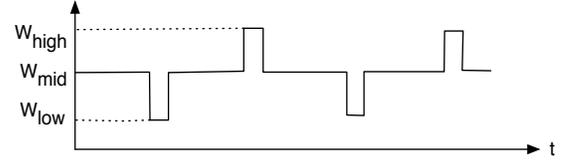}
\caption{Periodic $W(t)$ process in the example.}
\label{fig:two}
\end{figure}

To get some intuition behind the working of this algorithm, consider the following simple example. 
Suppose $W(t)$ can take three possible values from the set $\{W_{low}, W_{mid}, W_{high}\}$ where
$W_{low} < W_{mid} < W_{high}$. Similarly, $C(t)$ can take three possible values in $\{C_{low}, C_{mid}, C_{high}\}$
where $C_{low} < C_{mid} < C_{high}$ and does not depend on $P(t)$. We assume that the workload process
evolves in a frame-based periodic fashion. Specifically, in every odd numbered frame, $W(t) = W_{mid}$ for all except the
last slot of the frame when $W(t) = W_{low}$. In every even numbered frame, $W(t) = W_{mid}$ for all except the
last slot of the frame when $W(t) = W_{high}$. This is illustrated in Fig. \ref{fig:two}.
The $C(t)$ process evolves similarly, such that
$C(t) = C_{low}$ when $W(t) = W_{low}$, $C(t) = C_{mid}$ when $W(t) = W_{mid}$, and
$C(t) = C_{high}$ when $W(t) = W_{high}$. 

In the following, we assume a frame size of $5$ slots with $W_{low} = 10$, $W_{mid} = 15$, and $W_{high} = 20$ units. 
Also, $C_{low} = 2$, $C_{mid} = 6$, and $C_{high} = 10$ dollars. Finally, $R_{max} = D_{max} = 10$, $P_{peak} = 20$,
$C_{rc} = C_{dc} = 5$, $Y_{init} = Y_{min} = 0$ and we vary $Y_{max} > R_{max} + D_{max}$. In this example,
intuitively, an optimal algorithm that knows the workload and unit cost process beforehand would recharge the battery
as much as possible when $C(t) = C_{low}$ and discharge it as much as possible when $C(t) = C_{high}$. 
In fact, it can be shown that the following strategy is feasible and achieves minimum average cost:
\begin{itemize}
\item If $C(t) = C_{low}, W(t) = W_{low}$, then $P(t) = W_{low} + R_{max}$, $R(t) = R_{max}$, $D(t) = 0$. 
\item If $C(t) = C_{mid}, W(t) = W_{mid}$, then $P(t) = W_{mid}$, $R(t) = 0$, $D(t) = 0$. 
\item If $C(t) = C_{high}, W(t) = W_{high}$, then $P(t) = W_{high} - D_{max}$, $R(t) = 0$, $D(t) = D_{max}$. 
\end{itemize}
The time average cost resulting from this strategy can be easily calculated and is given by $87.0$ dollars/slot for all $Y_{max} > 10$.
Also, we note that the cost resulting from an algorithm that does not use the battery in this example is given by $94.0$ dollars/slot.

Now we simulate the dynamic algorithm for this example for different values of $Y_{max}$ for $1000$ slots ($200$ frames).
The value of $V$ is chosen to be $\frac{Y_{max} - Y_{min} - R_{max} - D_{max}}{C_{high} - C_{low}} = 
\frac{Y_{max} - 20}{8}$ (this choice will become clear in Sec. \ref{section:performance_theorem} when we relate $V$ to the
battery capacity).
Note that the number of slots for which a fully charged battery can sustain the data center at maximum load is
$Y_{max}/W_{high}$.

\begin{table}
\centering
 \begin{tabular}{|c|c|c|c|c|c|l|} \hline
   $Y_{max}$ & 20   & 30   & 40   & 50    & 75     & 100  \\ \hline
   $V$       &  0   & 1.25 & 2.5  & 3.75  & 6.875  & 10.0  \\ \hline
   Avg. Cost & 94.0 & 92.5 & 91.1 & 88.5  & 88.0   & 87.0 \\ \hline
  \end{tabular}
\caption{Average Cost vs. $Y_{max}$}
\label{table:table1}
\end{table}

In Table \ref{table:table1}, we show the time average cost achieved for different values of $Y_{max}$. It can be seen that
as $Y_{max}$ increases, the time average cost approaches the optimal value. (This behavior will be formalized in Theorem \ref{thm:one})
This is remarkable given that the dynamic algorithm operates without any knowledge of the future workload and cost processes. 
To examine the behavior of the dynamic algorithm in more detail, 
we fix $Y_{max} = 100$ and look at the sample paths of the control decisions taken by the
optimal offline algorithm and the dynamic algorithm
during the first $200$ slots. This is shown in Figs. \ref{fig:example2} and \ref{fig:example3}.
It can be seen that initially, the dynamic tends to perform suboptimally. But eventually it 
\emph{learns} to make close to optimal decisions.

\begin{figure}
\centering
\includegraphics[height=3cm, width=9cm,angle=0]{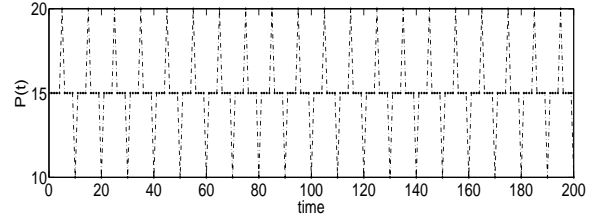}
\caption{$P(t)$ under the offline optimal solution with $Y_{max} = 100$.}
\label{fig:example2}
\end{figure}

\begin{figure}
\centering
\includegraphics[height=3cm, width=9cm,angle=0]{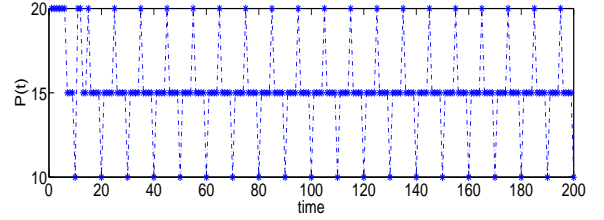}
\caption{$P(t)$ under the Dynamic Algorithm with $Y_{max} = 100$.}
\label{fig:example3}
\end{figure}


It might be tempting to conclude from this example that an algorithm based on a price threshold is optimal.
Specifically, such an algorithm makes a recharge vs. discharge decision depending on whether the current price 
$C(t)$ is smaller or larger than a threshold. However, it is easy to construct examples where the dynamic
algorithm outperforms such a threshold based algorithm. Specifically, suppose that the $W(t)$ process takes values 
from the interval $[10, 90]$ uniformly at random every slot. Also, suppose $C(t)$ takes values from the set
$\{2, 6, 10\}$ dollars uniformly at random every slot. We fix the other parameters as follows:
$R_{max} = D_{max} = 10$, $P_{peak} = 90$, $C_{rc} = C_{dc} = 1$, $Y_{init} = Y_{min} = 0$ and $Y_{max} = 100$. 
We then simulate a threshold based algorithm for different values of the threshold in the set 
$\{2, 6, 10\}$  and select the one that yields the smallest cost. This was found to be $280.7$ dollars/slot. We then simulate the
dynamic algorithm for $10000$ slots with $V = \frac{Y_{max} - 20}{10 - 2} = 10.0$ and it yields an average cost of $275.5$ dollars/slot. 
We also note that the cost resulting from an algorithm that does not use the battery in this example is given by $300.73$ dollars/slot.




We now establish two properties of the structure of the optimal solution to \textbf{P3} that will be useful
in analyzing its performance later.

\begin{lem} The optimal solution to \textbf{P3} has the following properties:
\begin{enumerate}
\item If $X(t) > -V C_{min}$, then the optimal solution always chooses $R^*(t) = 0$.
\item If $X(t) < -V \chi_{min}$, then the optimal solution always chooses $D^*(t) = 0$.
\end{enumerate}
\label{lem:two}
\end{lem}

\begin{proof} See Appendix A.
\end{proof}

\subsection{Solving \textbf{P3}}
\label{section:solving_p3}

In general, the complexity of solving \textbf{P3} depends on the structure of the unit cost function $C(t)$. For many cases of practical interest,
\textbf{P3} is easy to solve and admits closed form solutions that can be implemented in real time. We consider two such cases here.
Let $\theta(t)$ denote the value of the objective in \textbf{P3} 
when there is no recharge or discharge. Thus $\theta(t) = W(t)(X(t) + V C(t))$.

\subsubsection{$C(t)$ does not depend on $P(t)$}
\label{section:case1}

Suppose that $C(t)$ depends only on $S(t)$ and not on $P(t)$. 
We can rewrite the expression in the objective of \textbf{P3} as $P(t)(X(t) + V C(t)) + 1_R(t) V C_{rc} + 1_D(t) V C_{dc}$. 
Then, the optimal solution has the following simple threshold structure. 
\begin{enumerate}
\item If $X(t) + VC(t) > 0$, then $R^*(t) = 0$ so that there is no recharge and we have the following two cases:
\begin{enumerate}
\item If $P_{low}(X(t) + VC(t)) + V C_{dc} < \theta(t)$, then discharge as much as possible, so 
that we get $D^*(t) = \min[W(t), D_{max}]$, $P^*(t) = \max[0, W(t) - D_{max}]$.
\item Else, draw all power from the grid. This yields $D^*(t) = 0$ and $P^*(t) = W(t)$.
\end{enumerate}
\item Else if $X(t) + VC(t) \leq 0$, then $D^*(t) = 0$ so that there is no discharge and we have the following two cases:
\begin{enumerate}
\item If $P_{high}(X(t) + VC(t)) + V C_{rc} < \theta(t)$, then recharge as much as possible. This yields
 $R^*(t) = \min[P_{peak} - W(t), R_{max}]$ and $P^*(t) = \min[P_{peak}, W(t) + R_{max}]$.
\item Else, draw all power from the grid. This yields $R^*(t) = 0$ and $P^*(t) = W(t)$.
\end{enumerate}
\end{enumerate}



We will show that this solution is feasible and does not violate the finite battery constraint in Sec. \ref{section:performance_analysis}.

\subsubsection{$C(t)$ convex, increasing in $P(t)$}
\label{section:case2}

Next suppose for each $S(t)$, $C(t)$ is convex and increasing in $P(t)$. For example, $\hat{C}(S(t), P(t))$ may have the form
$\alpha(S(t))P^2(t)$ where $\alpha(S(t)) > 0$ for all $S(t)$. 
In this case, \textbf{P3} becomes a standard convex optimization problem in a single variable $P(t)$ and
can be solved efficiently. The full solution is provided in Appendix E.

\subsection{Performance Theorem}
\label{section:performance_theorem}

We first define an upper bound $V_{max}$ on the maximum value that $V$ can take in our algorithm.
\begin{align}
V_{max} \defequiv \frac{Y_{max} - Y_{min} - R_{max} - D_{max}}{\chi_{min} - C_{min}}
\label{eq:v_max}
\end{align}
Then we have the following result.
\begin{thm} (Algorithm Performance) Suppose the initial battery charge level $Y_{init}$ satisfies
$Y_{min} \leq Y_{init} \leq Y_{max}$. Then implementing the algorithm above 
with any fixed parameter $V$ such that $0 < V \leq V_{max}$ for all $t \in \{0, 1, 2, \ldots\}$ results in the
following performance guarantees:
\begin{enumerate}
\item The queue $X(t)$ is deterministically upper and lower bounded for all $t$ as follows:
\begin{align}
-V \chi_{min} - D_{max} \leq X(t) &\leq  Y_{max} - Y_{min} \nonumber \\
&\;\;\;\; - D_{max}  - V \chi_{min}
\label{eq:x_bounds}
\end{align}
\item The actual battery level $Y(t)$ satisfies $Y_{min} \leq Y(t) \leq Y_{max}$ for all $t$.
\item  All control decisions are feasible.
\item If $W(t)$ and $S(t)$ are i.i.d. over slots, then the time-average cost under the dynamic algorithm is within $B/V$ of the optimal value:
\begin{align}
\lim_{t\rightarrow\infty} \frac{1}{t} \sum_{\tau = 0}^{t-1}& \expect{P(\tau) C(\tau) + 1_R(\tau) C_{rc} + 1_D(\tau) C_{dc}} \nonumber \\
&\qquad \qquad \leq \phi_{opt} + B/V
\label{eq:utility_bound}
\end{align}
where $B$ is a constant given by $B = \frac{\max[R^2_{max}, D^2_{max}]}{2}$ 
and $\phi_{opt}$ is the optimal solution to \textbf{P1} under \emph{any} feasible control algorithm (possibly
with knowledge of future events). 
\end{enumerate}
\label{thm:one}
\end{thm}
Theorem \ref{thm:one} part $4$ shows that by choosing larger $V$, the time-average cost under the dynamic algorithm can be pushed
 closer to the minimum possible value $\phi_{opt}$. However, $V_{max}$ limits how large $V$ can be chosen.


\subsection{Proof of Theorem \ref{thm:one}}
\label{section:performance_analysis}

Here we prove Theorem \ref{thm:one}. 
\begin{proof} (Theorem \ref{thm:one} part $1$) 
We first show that (\ref{eq:x_bounds}) holds for $t=0$. We have that 
\begin{align}
Y_{min} \leq Y(0) = Y_{init} \leq Y_{max}
\label{eq:Y_0}
\end{align}
Using the definition  
(\ref{eq:x_t_def}), we have that $Y(0) = X(0) + V \chi_{min} + D_{max} + Y_{min}$. Using this in (\ref{eq:Y_0}), we get:
\begin{align*}
Y_{min} \leq X(0) + V \chi_{min} + D_{max} + Y_{min} \leq Y_{max}
\end{align*}
This yields 
\begin{align*}
-V \chi_{min} - D_{max} \leq X(0) &\leq Y_{max} - Y_{min} - D_{max} \\
&\;\;\;\; - V \chi_{min}
\end{align*}

Now suppose (\ref{eq:x_bounds}) holds for slot $t$.  We will show that it also holds for slot $t+1$. First, suppose $-V C_{min} < X(t) \leq 
Y_{max} - Y_{min} - D_{max} - V \chi_{min}$. Then,
from Lemma \ref{lem:two}, we have that $R^*(t) = 0$. Thus, using (\ref{eq:x_t}), we have that $X(t+1) \leq X(t) \leq 
Y_{max} - Y_{min} - D_{max} - V \chi_{min}$. 
Next, suppose $X(t) \leq -V C_{min}$. Then, the maximum possible increase is $R_{max}$ so that $X(t+1) \leq -V C_{min} + R_{max}$. 
Now for all $V$ such that $0 < V \leq V_{max}$, we have that 
$-V C_{min} + R_{max} \leq Y_{max} - Y_{min} - D_{max} - V \chi_{min}$. 
This follows from the definition (\ref{eq:v_max}) and the fact that $\chi_{min} \geq C_{min}$. 
Thus, we have $X(t+1) \leq Y_{max} - Y_{min} - D_{max} - V \chi_{min}$.


Next, suppose $-V\chi_{min} - D_{max} \leq X(t) < -V \chi_{min}$. Then, from Lemma \ref{lem:two}, we have that $D^*(t) = 0$. 
Thus, using (\ref{eq:x_t}) we have that $X(t+1) \geq X(t) \geq -V \chi_{min} - D_{max}$. Next, suppose $-V \chi_{min} \leq X(t)$. 
Then the maximum possible decrease is $D_{max}$ so that $X(t+1) \geq -V \chi_{min} - D_{max}$ for this case as well. This shows that
$X(t+1) \geq -V \chi_{min} - D_{max}$. Combining these two bounds proves (\ref{eq:x_bounds}).
\end{proof}
\begin{proof} (Theorem \ref{thm:one} parts $2$ and $3$) Part $2$ directly follows from (\ref{eq:x_bounds}) and (\ref{eq:x_t_def}). Using
$Y(t) = X(t) + V \chi_{min} + D_{max} + Y_{min}$ in the lower bound in (\ref{eq:x_bounds}), we have:
$-V \chi_{min} - D_{max} \leq Y(t) - V \chi_{min} - D_{max} - Y_{min}$, i.e.,  
$Y_{min} \leq Y(t)$. 
Similarly, using $Y(t) = X(t) + V \chi_{min} + D_{max} + Y_{min}$ in the upper bound in (\ref{eq:x_bounds}), we have:
$Y(t) - V \chi_{min} - D_{max} - Y_{min} \leq Y_{max} - Y_{min} - D_{max} - V \chi_{min}$, i.e.,
$Y(t) \leq Y_{max}$.

Part $3$ now follows from part $2$ and the constraint on $P(t)$ in \textbf{P3}. 
\end{proof}
\begin{proof} (Theorem \ref{thm:one} part $4$) We make use of the technique of Lyapunov optimization 
to show (\ref{eq:utility_bound}).
We start by defining a Lyapunov function as a scalar measure of congestion in the system.
Specifically, we define the following Lyapunov function: $L(X(t)) \defequiv \frac{1}{2} X^2(t)$. 
Define the conditional $1$-slot Lyapunov drift as follows:
\begin{align}
\Delta(X(t)) \defequiv \expect{L(X(t+1)) - L(X(t))|X(t)}
\label{eq:drift1}
\end{align}
Using (\ref{eq:x_t}), $\Delta(X(t))$ can be bounded as follows (see Appendix B for details):
\begin{align}
\Delta(X(t)) \leq B - X(t)\expect{D(t) - R(t)|X(t)}
\label{eq:drift2}
\end{align}
where $B = \frac{\max[R^2_{max}, D^2_{max}]}{2}$. Following the Lyapunov optimization framework, we add to
both sides of (\ref{eq:drift2}) the penalty term $V \expect{P(t)C(t) + 1_R(t)C_{rc} + 1_D(t)C_{dc}|X(t)}$ to
get the following:
\begin{align}
&\Delta(X(t))+ V \expect{P(t)C(t) + 1_R(t)C_{rc} + 1_D(t)C_{dc}|X(t)}  \nonumber \\
&\qquad \leq B - X(t)\expect{D(t) - R(t)|X(t)} \nonumber\\ 
&\qquad\;\;+V \expect{P(t)C(t) + 1_R(t)C_{rc} + 1_D(t)C_{dc}|X(t)} 
\label{eq:drift3}
\end{align}
Using the relation $W(t) = P(t) - R(t) + D(t)$, we can rewrite the above as:
\begin{align}
&\Delta(X(t))+ V \expect{P(t)C(t) + 1_R(t)C_{rc} + 1_D(t)C_{dc}|X(t)}  \leq \nonumber \\
&\qquad B - X(t)\expect{W(t)|X(t)} + X(t)\expect{P(t)|X(t)} \nonumber\\ 
&\qquad + V \expect{P(t)C(t) + 1_R(t)C_{rc} + 1_D(t)C_{dc}|X(t)} 
\label{eq:drift4}
\end{align}
Comparing this with \textbf{P3}, it can be seen that given any queue value $X(t)$,
our control algorithm is designed to \emph{minimize} the
right hand side of (\ref{eq:drift4}) over all possible feasible control policies. This includes the
optimal, stationary, randomized policy given in Lemma $\ref{lem:one}$.
Then, plugging the control decisions corresponding to the stationary, randomized policy, it can be shown that:
\begin{align*}
&\Delta(X(t))+ V \expect{P(t)C(t) + 1_R(t)C_{rc} + 1_D(t)C_{dc}|X(t)}  \leq \nonumber \\
& B + V \expect{P^{stat}(t)C^{stat}(t) + 1^{stat}_R(t)C_{rc} + 1^{stat}_D(t)C_{dc}|X(t)} \nonumber \\
& = B + V \hat{\phi} \leq B + V \phi_{opt}
\end{align*}
Taking the expectation of both sides and using the law of iterated expectations and summing over $t \in \{0, 1, 2, \ldots, T-1\}$, we get:
\begin{align*}
&\sum_{t=0}^{T-1} V \expect{P(t)C(t) + 1_R(t)C_{rc} + 1_D(t)C_{dc}} \leq \\
& BT + VT \phi_{opt} - \expect{L(X(T))} + \expect{L(X(0))}
\end{align*}
Dividing both sides by $VT$ and taking limit as $T \rightarrow\infty$ yields:
\begin{align*}
\lim_{T\rightarrow\infty} \frac{1}{T} \sum_{t = 0}^{T-1}& \expect{P(t) C(t) + 1_R(t) C_{rc} + 1_D(t) C_{dc}} \leq \phi_{opt} + B/V
\end{align*}
where we have used the fact that $\expect{L(X(0))}$ is finite and that $\expect{L(X(T))}$ is non-negative.
\end{proof}

\section{Extensions to Basic Model}
\label{section:ext_model}

In this section, we extend the basic model of Sec. \ref{section:basic_model} to the case where portions 
of the workload are delay-tolerant in the sense they can be postponed by a certain amount without affecting 
the utility the data center derives from executing them. We refer to such postponement as buffering the workload.
Specifically, we assume that the total workload consists of both delay tolerant and delay intolerant components. Similar to the workload in the 
basic model, the delay intolerant workload cannot be buffered and must be served immediately. However, the delay tolerant component may be
buffered and served later. As an example, data centers run virus scanning programs on most of their servers routinely 
(say once per day). As long as a virus scan is executed once a day, their purpose is served - it does not matter what time of the 
day is chosen for this. 
The ability to delay some of the workload gives more opportunities to reduce the average power cost in addition to using the battery.
We assume that our data center has system mechanisms to implement such buffering of specified workloads.


In the following, we will denote the total workload generated in slot $t$ by $W(t)$. This consists of the delay tolerant and intolerant components 
denoted by $W_1(t)$ and $W_2(t)$ respectively, so that $W(t) = W_1(t) + W_2(t)$ for all $t$. Similar to the basic model, we use $P(t), R(t), D(t)$
to denote the total power drawn from the grid, the total power used to recharge the battery and the total power discharged from the battery in slot $t$,
respectively. Thus, the total amount available to serve the workload is given by $P(t) - R(t) + D(t)$. Let $\gamma(t)$ denote the fraction of this
that is used to serve the delay tolerant workload in slot $t$.
Then the amount used to serve the delay intolerant workload is $(1-\gamma(t))(P(t) - R(t) + D(t))$.
Note that the following constraint must be satisfied every slot:
\begin{align}
0 \leq \gamma(t) \leq 1
\label{eq:gamma_t}
\end{align}

\begin{figure}
\centering
\includegraphics[height=4cm, width=7.5cm,angle=0]{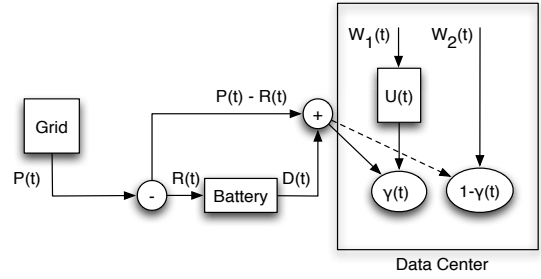}
\caption{Block diagram for the extended model with delay tolerant and delay intolerant workloads.}
\label{fig:ext_model}
\end{figure}

We next define $U(t)$ as the unfinished work for the delay tolerant workload in slot $t$. The dynamics for $U(t)$ can be expressed as:
\begin{align}
U(t+1) = \max[U(t) - \gamma(t)(P(t) - R(t) + D(t)), 0] + W_1(t)
\label{eq:u_t}
\end{align}
For the delay intolerant workload, there are no such queues since all incoming workload must be served in the same slot. This means:
\begin{align}
W_2(t) = (1- \gamma(t))(P(t) - R(t) + D(t))
\label{eq:w2_t}
\end{align}
The block diagram for this extended model is shown in Fig. \ref{fig:ext_model}.
Similar to the basic model, we assume that for $i=1, 2$, $W_i(t)$ varies randomly in an i.i.d. fashion, taking values from a set $\mathcal{W}_i$ of
non-negative values. 
We assume that $W_1(t) + W_2(t) \leq W_{max}$ for all $t$. 
We also assume that $W_1(t) \leq W_{1,max} < W_{max}$ and $W_2(t) \leq W_{2, max} < W_{max}$ for all $t$.
We further assume that $P_{peak} \geq W_{max} + \max[R_{max}, D_{max}]$. 
We use the same model for battery and unit cost as in Sec. \ref{section:basic_model}.

Our objective is to minimize the time-average cost subject to meeting all the constraints (such as finite battery size and (\ref{eq:w2_t}))
and ensuring finite average delay for the delay tolerant workload. This can be stated as:
\begin{align}
& \textrm{\bf{P4}}: \nonumber\\
&\textrm{Minimize:} \;\;  \lim_{t\rightarrow\infty} \frac{1}{t} \sum_{\tau = 0}^{t-1} \expect{P(\tau) C(\tau) + 1_R(\tau) C_{rc} + 1_D(\tau) C_{dc}} \nonumber\\
&\textrm{Subject to:}\;\; \textrm{Constraints} \; (\ref{eq:rt_dt}), (\ref{eq:rt_dt_bound}), (\ref{eq:rt_bound}), (\ref{eq:dt_bound}), (\ref{eq:p_peak}), 
(\ref{eq:gamma_t}), (\ref{eq:w2_t}) \nonumber\\
& \qquad \qquad \;\;\;\;\; \textrm{Finite average delay for $W_1(t)$} \nonumber
\end{align}
Similar to the basic model, we consider the following \emph{relaxed} problem:
\begin{align}
& \textrm{\bf{P5}}: \nonumber\\
&\textrm{Minimize:} \;\;  \lim_{t\rightarrow\infty} \frac{1}{t} \sum_{\tau = 0}^{t-1} \expect{P(\tau) C(\tau) + 1_R(\tau) C_{rc} + 1_D(\tau) C_{dc}} \nonumber\\
&\textrm{Subject to:}\;\; \textrm{Constraints} \; (\ref{eq:rt_dt}), (\ref{eq:rt_dt_bound}), (\ref{eq:p_peak}), (\ref{eq:gamma_t}), (\ref{eq:w2_t}) \nonumber\\
& \qquad \qquad \;\;\;\;\; \overline{R} = \overline{D} \label{eq:27} \\
& \qquad \qquad \;\;\;\;\; \overline{U} < \infty \label{eq:28}
\end{align}
where $\overline{U}$ is the time average expected queue backlog for the delay tolerant workload and is defined as:
\begin{align}
\overline{U} \defequiv \limsup_{t\rightarrow \infty} \frac{1}{t} \sum_{\tau=0}^{t-1} \expect{U(\tau)}
\label{eq:u_bar}
\end{align}
Let ${\phi}_{ext}$ and $\hat{\phi}_{ext}$ denote the optimal value for problems \textbf{P4} and \textbf{P5} respectively.
Since \textbf{P5} is less constrained than \textbf{P4}, we have that $\hat{\phi}_{ext} \leq {\phi}_{ext}$.
Similar to Lemma $\ref{lem:one}$, the following holds:

\begin{lem} (Optimal Stationary, Randomized Policy): If the workload process $W_1(t), W_2(t)$ and auxiliary process $S(t)$ are i.i.d. over slots, then there
exists a stationary, randomized policy that takes control decisions $\hat{P}(t), \hat{R}(t), \hat{D}(t), \hat{\gamma}(t)$
every slot purely as a function (possibly randomized) of the current state $(W_1(t), W_2(t), S(t))$ while satisfying the 
constraints (\ref{eq:w2_t}), (\ref{eq:rt_dt}), (\ref{eq:rt_dt_bound}), (\ref{eq:p_peak}), (\ref{eq:gamma_t}) and providing the following guarantees:
\begin{align}
&\expect{\hat{R}(t)} = \expect{\hat{D}(t)} \label{eq:stat1_ext} \\
&\expect{\hat{\gamma}(t)(\hat{P}(t) - \hat{R}(t) + \hat{D}(t))} \geq \expect{W_1(t)} \label{eq:stat2_ext}\\
&\expect{\hat{P}(t) \hat{C}(t) + \hat{1}_{{R}}(t) C_{rc} + \hat{1}_{{D}}(t) C_{dc}} = \hat{\phi}_{ext} \label{eq:stat3_ext}
\end{align}
where the expectations above are with respect to the stationary distribution of $(W_1(t), W_2(t), S(t))$ and the randomized control decisions.
\label{lem:three}
\end{lem}
\begin{proof} This result follows from the framework in \cite{neely-NOW, neely_new} and is omitted for brevity.
\end{proof}
The condition (\ref{eq:stat2_ext}) only guarantees queueing stability, not bounded worst case delay. 
We will now design a dynamic control algorithm that will yield bounded worst case delay while guaranteeing average cost that is within
$O(1/V)$ of $\hat{\phi}_{ext}$.

\subsection{Delay-Aware Queue}
\label{section:delay_aware_queue}

In order to provide worst case delay guarantees to the delay tolerant workload, we will make use of the technique of
$\epsilon$-persistent queue \cite{neely_smartgrid}. Specifically, we define a virtual queue $Z(t)$ as follows:
\begin{align}
Z(t+1) = [&Z(t) - \gamma(t)(P(t) - R(t) + D(t)) + \epsilon 1_{U(t)}]^+ 
\label{eq:z_t}
\end{align}
where $\epsilon > 0$ is a parameter to be specified later, $1_{U(t)}$ is an indicator variable that is $1$ if $U(t) > 0$ and $0$ else, and
$[x]^+ = \max[x, 0]$. The objective of this virtual queue is to enable the provision of worst-case
delay guarantee on any buffered workload $W_1(t)$. Specifically, if any control algorithm ensures that $U(t) \leq U_{max}$ and $Z(t) \leq Z_{max}$
for all $t$, then the worst case delay can be bounded. This is shown in the following:

\begin{lem} (Worst Case Delay)
Suppose a control algorithm ensures that $U(t) \leq U_{max}$ and $Z(t) \leq Z_{max}$
for all $t$, where $U_{max}$ and $Z_{max}$ are some positive constants. 
Then the worst case delay for any delay tolerant workload is at most $\delta_{max}$ slots where:
\begin{align}
\delta_{max} \defequiv \lceil (U_{max} + Z_{max})/\epsilon \rceil
\label{eq:delta_max}
\end{align}
\label{lem:worst_case_delay}
\end{lem}

\begin{proof} 
Consider a new arrival $W_1(t)$ in any slot $t$. We will show that this is served on or before time $t+\delta_{max}$.
We argue by contradiction. Suppose this workload is not served by $t+\delta_{max}$. Then for all slots $\tau \in \{t+1, t+2, \ldots, t + \delta_{max}\}$,
it must be the case that $U(\tau) > 0$ (else $W_1(t)$ would have been served before $\tau$). This implies that $1_{U(\tau)} = 1$ and using (\ref{eq:z_t}),
we have: 
\begin{align*}
Z(\tau + 1) \geq Z(\tau) - \gamma(\tau)(P(\tau) - R(\tau) + D(\tau)) + \epsilon
\end{align*}
Summing for all $\tau \in \{t+1, t+2, \ldots, t + \delta_{max}\}$, we get:
\begin{align*}
Z(t + \delta_{max} + 1) &- Z(t+1) \geq \delta_{max}\epsilon \\
&- \sum_{\tau = t+1}^{t + \delta_{max}} [\gamma(\tau)(P(\tau) - R(\tau) + D(\tau))]
\end{align*}
Using the fact that $Z(t + \delta_{max} + 1) \leq Z_{max}$ and $Z(t+1) \geq 0$, we get:
\begin{align}
\sum_{\tau = t+1}^{t + \delta_{max}} [\gamma(\tau)(P(\tau) - R(\tau) + D(\tau))] \geq \delta_{max}\epsilon - Z_{max}
\label{eq:32}
\end{align}
Note that by (\ref{eq:u_t}), $W_1(t)$ is part of the backlog $U(t+1)$. Since $U(t+1) \leq U_{max}$ and since the service is FIFO, it will be
served on or before time $t+\delta_{max}$ whenever at least $U_{max}$ units of power is used to serve the delay tolerant workload during the
interval $(t+1, \ldots, t+\delta_{max})$. Since we have assumed that $W_1(t)$ is not served by $t+\delta_{max}$, it must be the case that 
$\sum_{\tau = t+1}^{t + \delta_{max}} [\gamma(\tau)(P(\tau) - R(\tau) + D(\tau))] < U_{max}$. Using this in (\ref{eq:32}), we have:
\begin{align*}
U_{max} > \delta_{max}\epsilon - Z_{max}
\end{align*}
This implies that $\delta_{max} < (U_{max} + Z_{max})/\epsilon$, that contradicts the definition of $\delta_{max}$ in (\ref{eq:delta_max}).
\end{proof}

In Sec. \ref{section:performance_theorem_ext}, we will show that there are indeed constants $U_{max}, Z_{max}$
such that the dynamic algorithm ensures that $U(t) \leq U_{max}, Z(t) \leq Z_{max}$ for all $t$.

\subsection{Optimal Control Algorithm}
\label{section:optimal_algorithm_ext}

We now present an online control algorithm that approximately solves \textbf{P4}. 
Similar to the algorithm for the basic model, 
this algorithm also makes use of the following queueing state variable $X(t)$ to track the battery charge level and is defined as follows:
\begin{align}
X(t) = Y(t) - Q_{max} - D_{max} - Y_{min}
\label{eq:x_t_def_ext}
\end{align}
where $Q_{max}$ is a constant to be specified in (\ref{eq:q_max}).
Recall that $Y(t)$ denotes the actual battery charge level in slot $t$ and evolves according to
(\ref{eq:y_t}). It can be seen that $X(t)$ is simply a shifted version of $Y(t)$ and its dynamics is given by: 
\begin{align}
X(t+1) = X(t) - D(t) + R(t)
\label{eq:x_t_ext}
\end{align}
We will show that this definition enables our algorithm to ensure that the constraint (\ref{eq:y_t_bound}) is met. We are now ready to state the
dynamic control algorithm. Let $(W_1(t), W_2(t), S(t))$ be the system state in slot $t$.
Define $\boldsymbol{Q}(t) \defequiv (U(t), Z(t), X(t))$ as the queue state that includes the workload queue as well as auxiliary queues.
Then the dynamic algorithm chooses control decisions $P(t), R(t), D(t)$ and $\gamma(t)$ as the 
solution to the following optimization problem:
\begin{align}
& \textrm{\bf{P6}}: \nonumber\\
&\textrm{Max:} [U(t)+Z(t)]P(t) - V \Big[P(t) {C}(t) + 1_R(t)C_{rc} + 1_D(t) C_{dc}\Big] \nonumber\\
& \qquad \qquad \;\;\;\; + [X(t) + U(t) + Z(t)](D(t) - R(t)) \nonumber\\
&\textrm{Subject to:} \; \textrm{Constraints} \; (\ref{eq:gamma_t}) , (\ref{eq:w2_t}), (\ref{eq:rt_dt}), (\ref{eq:rt_dt_bound}), (\ref{eq:p_peak}) \nonumber 
\end{align}
where $V > 0$ is a control parameter that affects the distance from optimality.
Let $P^*(t), R^*(t), D^*(t)$ and $\gamma^*(t)$ denote the optimal solution to \textbf{P6}. 
Then, the dynamic algorithm allocates $(1 - \gamma^*(t))(P^*(t)-R^*(t)+ D^*(t))$ power to service the delay intolerant workload and
the remaining is used for the delay tolerant workload.


After computing these quantities, the algorithm implements them and updates the queueing variable $X(t)$ according to (\ref{eq:x_t_ext}). This process is
repeated every slot. Note that in solving \textbf{P6}, the control algorithm only makes use of the current system state values and does not require
knowledge of the statistics of the workload or unit cost processes. 

We now establish two properties of the structure of the optimal solution to \textbf{P6} that will be useful
in analyzing its performance later.
\begin{lem} The optimal solution to \textbf{P6} has the following properties:
\begin{enumerate}
\item If $X(t) > -V C_{min}$, then the optimal solution always chooses $R^*(t) = 0$.
\item If $X(t) < - Q_{max}$ (where $Q_{max}$ is specified in (\ref{eq:q_max})), then the optimal solution always chooses $D^*(t) = 0$.
\end{enumerate}
\label{lem:five}
\end{lem}
\begin{proof} 
See Appendix C.
\end{proof}

\subsection{Solving \textbf{P6}}
\label{section:solving_P6}

Similar to \textbf{P3}, the complexity of solving \textbf{P6} depends on the structure of the unit cost function $C(t)$. For many cases of practical interest,
\textbf{P6} is easy to solve and admits closed form solutions that can be implemented in real time. We consider one such case here.

\subsubsection{$C(t)$ does not depend on $P(t)$}
\label{section:case1_ext}

For notational convenience, let $Q_1(t) = [U(t) + Z(t) - VC(t)]$ and
$Q_2(t) = [X(t) + U(t) + Z(t)]$. 

Let $\theta_1(t)$ denote the optimal value of the objective in \textbf{P6} 
when there is no recharge or discharge. When $C(t)$ does not depend on $P(t)$, this can be calculated
as follows: If $U(t) + Z(t) \geq VC(t)$, then $\theta_1(t) = Q_1(t) P_{peak}$. Else,
$\theta_1(t) = Q_1(t) W_2(t)$.

Next, let $\theta_2(t)$ denote the optimal value of the objective in \textbf{P6}
when the option of recharge is chosen, so that $R(t) > 0, D(t) = 0$. This can be calculated as follows:
\begin{enumerate}
\item If $Q_1(t) \geq 0, Q_2(t) \geq 0$, then $\theta_2(t) = Q_1(t) P_{peak} - V C_{rc}$.
\item If $Q_1(t) \geq 0, Q_2(t) < 0$, then $\theta_2(t) = Q_1(t) P_{peak} - Q_2(t) R_{max} - V C_{rc}$.
\item If $Q_1(t) < 0, Q_2(t) \geq 0$, then $\theta_2(t) = Q_1(t) W_2(t) - V C_{rc}$.
\item If $Q_1(t) < 0, Q_2(t) < 0$, then we have two cases:
\begin{enumerate}
\item If $Q_1(t) \geq Q_2(t)$, then $\theta_2(t) = Q_1(t) (R_{max} + W_2(t)) - Q_2(t) R_{max} - V C_{rc}$.
\item If $Q_1(t) < Q_2(t)$, then $\theta_2(t) = Q_1(t) W_2(t) - V C_{rc}$.
\end{enumerate}
\end{enumerate}
Finally, let $\theta_3(t)$ denote the optimal value of the objective in \textbf{P6}
when when the option of discharge is chosen, so that $D(t) > 0, R(t) = 0$. This can be calculated as follows:
\begin{enumerate}
\item If $Q_1(t) \geq 0, Q_2(t) \geq 0$, then $\theta_3(t) = Q_1(t) P_{peak} + Q_2(t) D_{max} - V C_{dc}$.
\item If $Q_1(t) \geq 0, Q_2(t) < 0$, then $\theta_3(t) = Q_1(t) P_{peak} - V C_{dc}$.
\item If $Q_1(t) < 0, Q_2(t) \geq 0$, then $\theta_3(t) = Q_1(t) \max[0, W_2(t) - D_{max}] + Q_2(t) D_{max} - V C_{dc}$.
\item If $Q_1(t) < 0, Q_2(t) < 0$, then we have two cases:
\begin{enumerate}
\item If $Q_1(t) \leq Q_2(t)$, then $\theta_3(t) = Q_1(t) \max[0, W_2(t) - D_{max}] + Q_2(t) \min[W_2(t), D_{max}] - V C_{dc}$.
\item If $Q_1(t) > Q_2(t)$, then $\theta_3(t) = Q_1(t) W_2(t) - V C_{dc}$.
\end{enumerate}
\end{enumerate}
After computing $\theta_1(t), \theta_2(t), \theta_3(t)$, we pick the mode that yields the highest value of the objective and
implement the corresponding solution.

\subsection{Performance Theorem}
\label{section:performance_theorem_ext}

We define an upper bound $V^{max}_{ext}$ on the maximum value that $V$ can take in our algorithm for the extended model.
\begin{align}
V^{max}_{ext} \defequiv \frac{Y_{max} - Y_{min} - (R_{max} + D_{max} + W_{1,max} + \epsilon)}{\chi_{min} - C_{min}}
\label{eq:v_max_ext}
\end{align}
Then we have the following result. 

\begin{thm} (Algorithm Performance) Suppose $U(0) = 0$, $Z(0) = 0$ and 
the initial battery charge level $Y_{init}$ satisfies
$Y_{min} \leq Y_{init} \leq Y_{max}$. Then implementing the algorithm above 
with any fixed parameter $\epsilon \geq 0$ such that $\epsilon \leq W_{max} - W_{2,max}$
and a parameter $V$ such that $0 < V \leq V^{max}_{ext}$ for all $t \in \{0, 1, 2, \ldots\}$ results in the
following performance guarantees:
\begin{enumerate}
\item The queues $U(t)$ and $Z(t)$ are deterministically upper bounded by $U_{max}$ and $Z_{max}$ respectively for all $t$ where:
\begin{align}
&U_{max} \defequiv V \chi_{min} + W_{1,max} \label{eq:u_max} \\
&Z_{max} \defequiv V \chi_{min} + \epsilon \label{eq:z_max}
\end{align}
Further, the sum $U(t) + Z(t)$ is also deterministically upper bounded by $Q_{max}$ where
\begin{align}
Q_{max} \defequiv V \chi_{min} + W_{1,max} + \epsilon\label{eq:q_max}
\end{align}

\item The queue $X(t)$ is deterministically upper and lower bounded for all $t$ as follows:
\begin{align}
- Q_{max} - D_{max} \leq X(t) &\leq Y_{max} - Y_{min} - Q_{max} \nonumber \\
&\;\;\;- D_{max}
\label{eq:x_bounds_ext}
\end{align}

\item The actual battery level $Y(t)$ satisfies $Y_{min} \leq Y(t) \leq Y_{max}$ for all $t$.

\item  All control decisions are feasible.

\item The worst case delay experienced by any delay tolerant request is given by:
\begin{align}
\Big\lceil \frac{2V \chi_{min} + W_{1,max} + \epsilon}{\epsilon} \Big\rceil
\end{align}

\item If $W_1(t), W_2(t)$ and $S(t)$ are i.i.d. over slots, then
the time-average cost under the dynamic algorithm is within $B_{ext}/V$ of the optimal value:
\begin{align}
\lim_{t\rightarrow\infty} \frac{1}{t} \sum_{\tau = 0}^{t-1}& \expect{P(\tau) C(\tau) + 1_R(\tau) C_{rc} + 1_D(\tau) C_{dc}} \nonumber \\
&\qquad \qquad \leq \hat{\phi}_{ext} + B_{ext}/V
\label{eq:utility_bound_ext}
\end{align}
where $B_{ext}$ is a constant given by $B_{ext} = (P_{peak} + D_{max})^2 + \frac{(W_{1,max})^2 + \epsilon^2}{2} + B$ 
and $\hat{\phi}_{ext}$ is the optimal solution to \textbf{P4} under \emph{any} feasible control algorithm (possibly
with knowledge of future events). 
\end{enumerate}
\label{thm:two}
\end{thm}
Thus, by choosing larger $V$, the time-average cost under the dynamic algorithm can be pushed
closer to the minimum possible value $\phi_{opt}$. However, this increases the worst case delay bound yielding
a $O(1/V, V)$ utility-delay tradeoff. Also note that  $V^{max}_{ext}$ limits how large $V$ can be chosen.

\begin{proof} 
 See Appendix D.
\end{proof}

\begin{figure}
\centering
\includegraphics[width=8.5cm,angle=0]{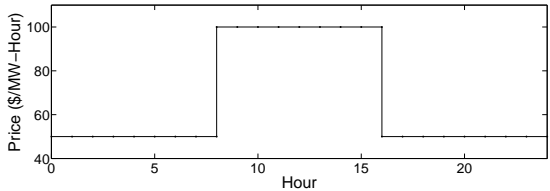}
\caption{One period of the unit cost process.}
\label{fig:c_periodic}
\end{figure}

\begin{figure}
\centering
\includegraphics[width=8.5cm,angle=0]{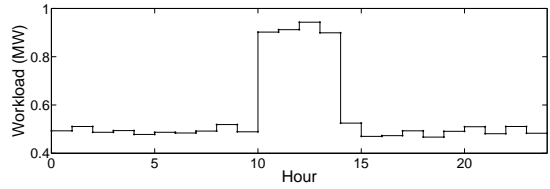}
\caption{One period of the workload process.}
\label{fig:w_periodic}
\end{figure}

\section{Simulation-based Evaluation}
\label{section:simulations}

We evaluate the performance of our control algorithms using both synthetic and
real pricing data. To gain insights into the behavior of our algorithms and to compare with the optimal offline solution, we first consider the basic model and
use a simple periodic unit cost and workload process as shown in Figs. \ref{fig:c_periodic} and \ref{fig:w_periodic}. 
These values repeat every $24$ hours and the unit cost does not depend on $P(t)$.
From Fig. \ref{fig:c_periodic}, it can be seen that $C_{max} = \$100$ and  $C_{min} = \$50$. Further, we have that $\chi_{min} = C_{max} = 100$. 
We assume a slot size of $1$ minute so that the
control decisions on $P(t), R(t), D(t)$ are taken once every minute. We fix the parameters $R_{max} = 0.2$ MW-slot, $D_{max} = 1.0$ MW-slot, $C_{rc} = C_{dc} = 0, Y_{min} = 0$.
We now simulate the basic control algorithm of Sec. \ref{section:case1} for different values of $Y_{max}$ and with $V = V_{max}$. For each $Y_{max}$, the simulation is
performed for a duration of $4$ weeks.

In Fig. \ref{fig:compare_periodic}, we plot the average cost per hour under the dynamic algorithm for different values of battery size $Y_{max}$. It can be seen that 
the average cost reduces as $Y_{max}$ is increased and converges to a fixed value for large $Y_{max}$, as suggested by Theorem \ref{thm:one}. For this simple example, we
can compute the minimum possible average cost per hour (over all battery sizes) 
and this is given by $\$33.23$ which is also the value to which the dynamic algorithm converges as $Y_{max}$ is increased.
Moreover, in this example, we can also compute the optimal offline cost for \emph{each value} of $Y_{max}$. These also also plotted in Fig. \ref{fig:compare_periodic}. 
It can be seen that, for each $Y_{max}$,
the dynamic algorithm performs quite close to the corresponding optimal value, 
even for smaller values of $Y_{max}$. Note that Theorem $1$ provides such guarantees only for sufficiently large
values of $Y_{max}$. Finally, the average cost per hour when no battery is used is given by $\$39.90$.

\begin{figure}
\centering
\includegraphics[width=8.5cm,angle=0]{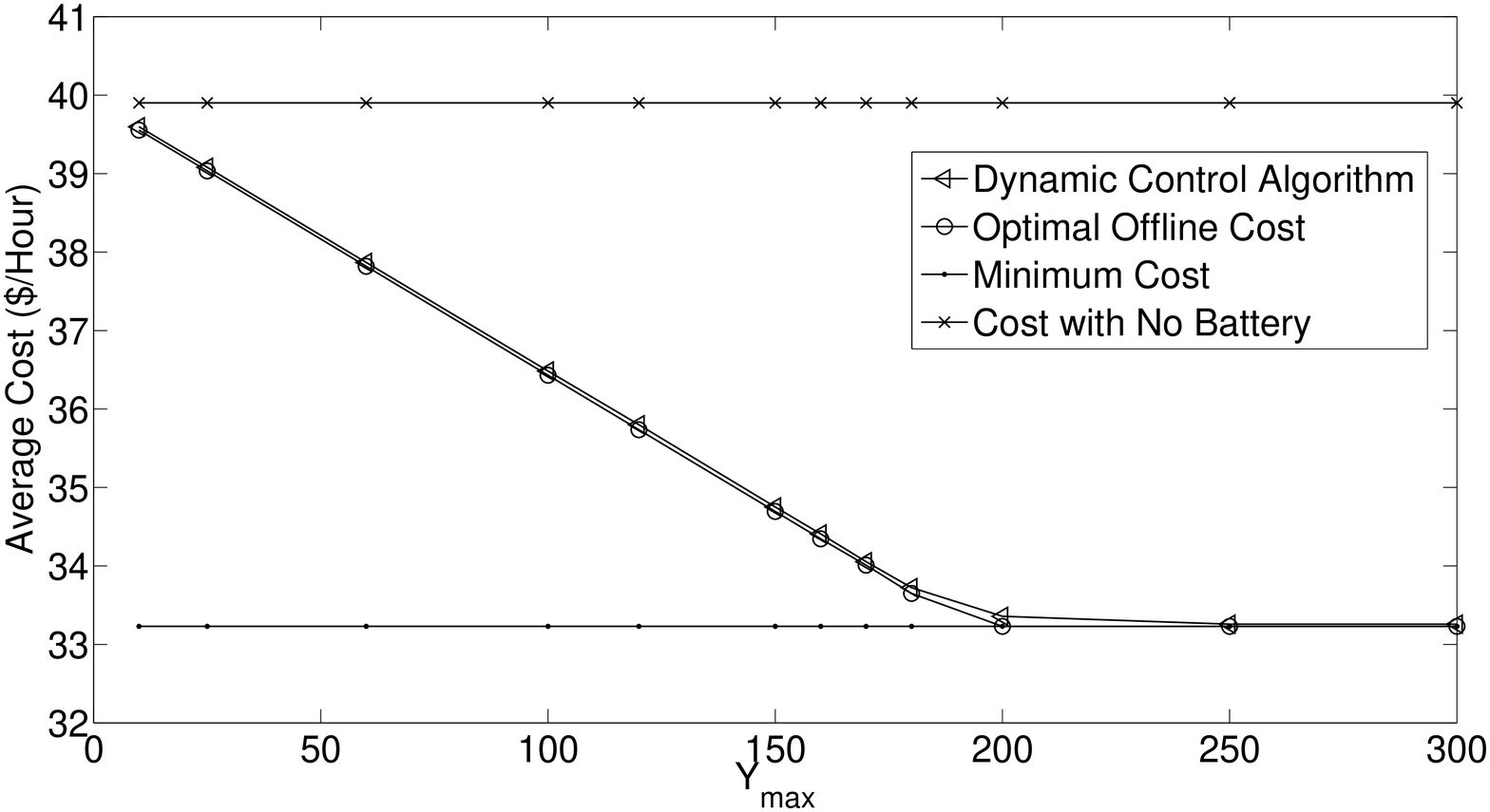}
\caption{Average Cost per Hour vs. $Y_{max}$.}
\label{fig:compare_periodic}
\end{figure}

We next consider a six-month data set of average hourly spot market prices
for the Los Angeles Zone LA1 obtained from CAISO~\cite{caiso}.
These prices correspond to the period 01/01/2005$-$06/30/2005 and
each value denotes the average price of 1 MW-Hour of electricity.
A portion of this data corresponding to the first week of January is plotted
in Fig.~\ref{fig:one_week}. We fix the slot size to 5 minutes so that 
control
decisions on $P(t), R(t), D(t)$, etc. are taken once every 5 minutes.
The unit cost $C(t)$ obtained from the data set for each hour is assumed
to be fixed for that hour. Furthermore, we assume that the unit cost does not
depend on the total power drawn $P(t)$.



In our experiments,
we assume that the data center receives workload in
an i.i.d fashion. Specifically, every slot, $W(t)$ takes values from the
set [0.1,1.5] MW uniformly at random. We fix the
parameters $D_{max}$ and $R_{max}$ to $0.5$ MW-slot, $C_{dc}=C_{rc}= \$0.1$, and $Y_{min} = 0$. Also, $P_{peak} = W_{max} + R_{max} = 2.0$ MW.
We now simulate four algorithms on this setup for different values of
$Y_{max}$. The length of time the battery can power the data center
if the draw were $W_{max}$ starting from fully charged battery is
given by $\frac{Y_{max}}{W_{max}}$ slots, each of length $5$ minutes. We consider
the following four control techniques: (A) ``No battery, No WP,'' which
meets the demand in every slot using power from the grid and without
postponing any workload, (B) ``Battery, No WP,'' which employs the
algorithm in the basic model without postponing any workload, (C) ``No Battery, WP,'' which employs the extended model for WP but without any battery,
and (D) ``Complete,'' the complete algorithm of the extended model
with both battery and WP. For (C) and (D), we assume that during every slot, half of the total
workload is delay-tolerant.

We simulate these algorithms to obtain 
the total cost over  the $6$ month period
for $Y_{max} \in \{15, 30, 50\}$ MW-slot. For (B), we use $V=V_{max}$ while for
(C) and (D), we use $V=V^{max}_{ext}$ with $\epsilon = W_{max}/2$.
Note that an increased battery capacity should have no effect on the performance under (C).
In order to get a fair comparison with the other schemes,
we assume that the worst case delay guarantee that case (C) must
provide for the delay tolerant traffic is the same as that under (D). 

Fig.~\ref{fig:compare_real_prices_iid_w} plots the total cost under these schemes over the $6$ month period.
In Table~\ref{table:table2}, we show the ratio of the total cost under schemes (B), (C), (D) to the total cost under (A)
for these values of $Y_{max}$ over the $6$ month period. It can be seen that
(D) provides the most cost savings over the baseline case.

\begin{figure}
\centering
\includegraphics[height=4.5cm, width=8.5cm,angle=0]{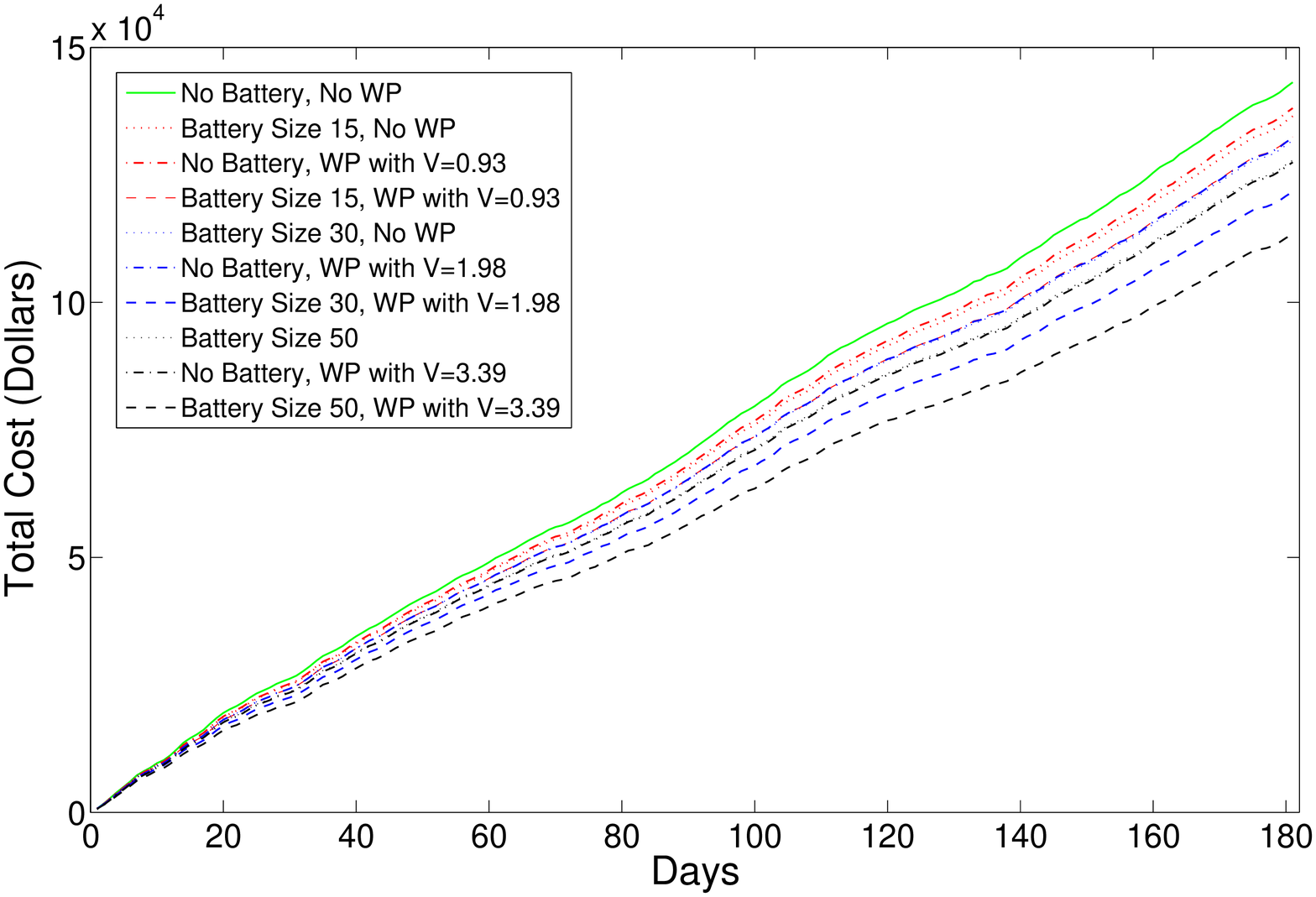}
\caption{Total Cost over $6$ months with i.i.d $W(t)$ and different $Y_{max}$}
\label{fig:compare_real_prices_iid_w}
\end{figure}


\begin{table}
\centering
 \begin{tabular}{|c|c|c|l|} \hline
   $Y_{max}$ & 15 & 30 & 50 \\ \hline
   Battery, No WP & 95\% & 92\% & 89\% \\ \hline
   WP, No Battery & 96\% & 92\% & 88\% \\ \hline
   WP, Battery & 92\% & 85\% & 79\%\\ \hline
  \end{tabular}
\caption{Ratio of cost under schemes (B), (C), (D) to the cost under (A)
for different values of $Y_{max}$ with i.i.d. $W(t)$.}
\label{table:table2}
\end{table}

\section{Conclusions and Future Work}
\label{section:conc}

In this paper, we studied the problem of 
opportunistically using energy storage devices 
to reduce the time average electricity bill of a data center. Using the technique of
Lyapunov optimization, we designed an online control algorithm that achieves close to
optimal cost as the battery size is increased.

We would like to extend our current framework along several important directions including: (i)
multiple utilities (or captive sources such as DG)
with different price variations and availability properties (e.g.,
certain renewable sources of energy are not available at all times), (ii) tariffs where the
utility bill depends on peak power draw in addition to the energy consumption, and (iii)
devising online algorithms that offer solutions whose proximity to the optimal has a smaller 
dependence on battery capacity than currently. We also plan to explore implementation
and feasibility related concerns such as: (i) what are appropriate trade-offs between investments
in additional battery capacity and cost reductions that this offers? (ii) what is the extent
of cost reduction benefits for realistic data center workloads? and  (iii) does stored energy
make sense as a cost optimization knob in other domains besides data centers?
Our technique could be viewed as a design tool which, when parameterized well,
can assist in  determining suitable configuration parameters such as
battery size, usage rules-of-thumb, time-scale at which decisions should be made, etc.
Finally, we believe that our work opens up a whole set of interesting issues worth exploring in the
area of consumer-end (not just data centers) demand response mechanisms for power cost
optimization.








\section*{Acknowledgments}
This work was supported in part by the NSF Career grant CCF-0747525.

\section*{Appendix A -Proof of Lemma \ref{lem:two}}

We can rewrite the expression in the objective of \textbf{P3} as $[X(t) + V C(t)]P(t) + V[1_R(t)C_{rc} + 1_D(t) C_{dc}]$. 
To show part $1$, suppose $R^*(t) = \delta > 0$ when $X(t) > -V C_{min}$, so that we have $P^*(t) - \delta = W(t)$, $D^*(t) = 0$, $1_R(t) = 1$, and $1_D(t) = 0$.
Then the value of the objective is given by:
\begin{align*}
&[X(t) + V C(P^*(t))]P^*(t) + V C_{rc} = \\
&[X(t) + V C(W(t)+\delta)] (W(t) + \delta) + V C_{rc} >\\
&[X(t) + V C(W(t))] W(t)
\end{align*}
where the last step follows by noting that $X(t) + V C(W(t)) > 0$ when $X(t) > -V C_{min}$ and that
$C(t)$ in non-negative and non-decreasing in $P(t)$. The last expression denotes the value of the objective
when $R^*(t) = D^*(t) = 0$ and all demand is met using power drawn from the grid and is smaller. 
This shows that when $X(t) > -V C_{min}$, then the optimal solution cannot choose $R^*(t) > 0$.



Next, to show part $2$,  suppose $D^*(t) = \delta > 0$ when $X(t) < -V\chi_{min}$, so that we have $P^*(t) + \delta = W(t)$, $R^*(t) = 0$, 
and $1_D(t) = 1$. Then the value of the objective is given by:  
\begin{align*}
&[X(t) + V C(P^*(t))]P^*(t) + V C_{dc} = \\
&[X(t) + V C(W(t)-\delta)](W(t) - \delta) + V C_{dc} \geq \\
&[X(t) + V C(W(t)-\delta)] (W(t) - \delta) >\\
&[X(t) + V C(W(t))] W(t)
\end{align*}
where in the last step, we used the property (\ref{eq:chi_defn}) together with the fact that $X(t) < -V\chi_{min}$.
The last expression denotes the value of the objective
when $R^*(t) = D^*(t) = 0$ and all demand is met using power drawn from the grid and is smaller.
This shows that when $X(t) < -V\chi_{min}$, then the optimal solution cannot choose $D^*(t) > 0$.

\section*{Appendix B - Proof of Bound (\ref{eq:drift2})}

Squaring both sides of (\ref{eq:x_t}), dividing by $2$, and rearranging yields:
\begin{align*}
\frac{X^2(t+1) - X^2(t)}{2} =  \frac{(D(t) - R(t))^2}{2} - X(t)[D(t) - R(t)]
\end{align*}
Now note that under any feasible algorithm, at most one of $R(t)$ and $D(t)$ can be non-zero. 
Further, since $R(t) \leq R_{max}, D(t) \leq D_{max}$ for all $t$, we have:
\begin{align*}
\frac{(D(t) - R(t))^2}{2}  \leq \frac{\max[R^2_{max}, D^2_{max}]}{2} = B
\end{align*}
Taking conditional expectations of both sides given $X(t)$, we have:
$\Delta(X(t)) \leq B - X(t)\expect{D(t) - R(t)|X(t)}$.


\section{Proof of Lemma 5}

Suppose $X(t) > -V C_{min}$ and $R^*(t) > 0, D^*(t) = 0$. 
Then, we have that
$W_2(t) = (1 - \gamma^*(t))(P^*(t) - R^*(t))$.
After rearranging, the value of the objective of \textbf{P6} can be expressed as:
\begin{align*}
[U(t) + Z(t)](&P^*(t) - R^*(t)) - VP^*(t)C(P^*(t)) - V C_{rc} \\
- X(t)R^*(t) &< [U(t) + Z(t)](P^*(t) - R^*(t)) \\
&- V[P^*(t) - R^*(t)]C(P^*(t) - R^*(t)) 
\end{align*}
where we used the inequalities 
$P^*(t) C(P^*(t) - R^*(t)) < P^*(t)C(P^*(t))$ 
and $X(t) + V C(P^*(t) - R^*(t)) > 0$.
The first follows from the non-negative and non-decreasing property of $C(t)$ in
$P(t)$. The second follows by noting that $X(t) > -V C_{min}$.
Now note that the right hand side denotes the value of the objective when power
$P(t) = P^*(t) - R^*(t)$ is drawn from the grid and the battery is not recharged or
discharged. This is a feasible option since by choosing $\gamma(t) = 0$, we have that
$W_2(t) = (P^*(t) - R^*(t))$.
This shows that when $X(t) > 0$, $R^*(t) > 0$ is not optimal.
This shows part $1$.

Next, suppose $X(t) < - Q_{max} < -V \chi_{min}$ and $D^*(t) > 0$, 
$R^*(t) = 0$.
Then, we have that
$W_2(t) = (1 - \gamma^*(t))(P^*(t) + D^*(t))$.
We consider two cases:\\

(1) $P^*(t) + D^*(t) \leq P_{peak}$: After  rearranging,
the value of the objective of \textbf{P6} can be expressed as:
\begin{align*}
[U(t) + Z(t)](&P^*(t) + D^*(t)) - VP^*(t)C(P^*(t)) - V C_{dc}  \\
+  X(t)D^*(t) &< [U(t) + Z(t)](P^*(t) + D^*(t)) \\
&- VP^*(t)C(P^*(t)) - V \chi_{min} D^*(t)
\end{align*}
where we used the fact that $X(t) < -V \chi_{min}$ and $D^*(t) > 0$. 
Using the property (\ref{eq:chi_defn}), we have:
\begin{align*}
(P^*(t) + D^*(t)) C(P^*(t) + D^*(t)) - P^*(t)C(P^*(t)) \\ \leq \chi_{min} D^*(t)
\end{align*}
Using this in the inequality above, we have:
\begin{align*}
[U(t) + Z(t)]&(P^*(t) + D^*(t)) - VP^*(t)C(P^*(t)) - V C_{dc}  \\
+ X(t)D^*(t) < &\; [U(t) + Z(t)](P^*(t) + D^*(t))\\
& - V(P^*(t) + D^*(t)) C(P^*(t) + D^*(t))
\end{align*}
Note that the last term denotes the value of objective when power
$P(t) = P^*(t) + D^*(t) \leq P_{peak}$ is drawn from the grid and the battery is not recharged or
discharged. This is a feasible option since by choosing 
$\gamma(t) = 0$, we have that
$W_2(t) = (P^*(t) + D^*(t))$.
This shows that for this case, when $X(t) < -V \chi_{min}$, $D^*(t) > 0$ is not optimal. \\

(2) $P^*(t) + D^*(t) > P_{peak}$: The value of the objective of \textbf{P6} is given by:
\begin{align*}
&[U(t) + Z(t)]P^*(t) - VP^*(t)C(P^*(t)) - V C_{dc}  \\
&+ [U(t) + Z(t) + X(t)]D^*(t) < [U(t) + Z(t)]P^*(t) \\
& \qquad \qquad - VP^*(t)C(P^*(t))
\end{align*}
where we used the fact that since $X(t) < - Q_{max}$ and
$U(t) + Z(t) \leq Q_{max}$ (Theorem \ref{thm:two} part $1$),
we have $[U(t) + Z(t) + X(t)] D^*(t) < 0$.
The last term in the inequality above 
denotes the value of the objective when power
$P(t) = P^*(t)$ is drawn from the grid and the battery is not recharged or
discharged. To see that this is feasible, note that we need
$(1 - {\gamma}(t)) P^*(t) = W_2(t)$ where ${\gamma}(t)$ must be $\leq 1$. Since
$W_2(t) \leq W_{2,max}$, this implies:
\begin{align*}
1 - {\gamma}(t) = \frac{W_2(t)}{P^*(t)} \leq \frac{W_{2,max}}{P^*(t)} \leq \frac{W_{2,max}}{P_{peak} - D_{max}} 
\end{align*}
where we used the fact that  $P^*(t) \geq P_{peak} - D^*(t) \geq P_{peak} - D_{max}$. Now since
$W_{2,max} \leq P_{peak} - D_{max}$, the last term above is $\leq 1$, so that choosing
$P(t) = P^*(t)$ and $D(t) = 0$ is a feasible option. This shows that for this case as well,
$D^*(t) > 0$ is not optimal.


\section*{Appendix D - Proof of Theorem \ref{thm:two}, parts 2-6}

Here, we prove parts $1-6$ of Theorem \ref{thm:two}.

\begin{proof} (Theorem \ref{thm:two} part $1$) We first show (\ref{eq:u_max}).
Clearly, (\ref{eq:u_max}) holds for $t=0$.
Now suppose it holds for slot $t$.  We will show that it also holds for slot $t+1$.
First suppose $U(t) \leq V \chi_{min}$. Then, by (\ref{eq:u_t}),
the most that $U(t)$ can increase in one slot is $W_{1,max}$ so that $U(t+1) \leq V \chi_{min} + W_{1,max}$. Next,
suppose $V \chi_{min} < U(t) \leq V \chi_{min} + W_{1,max}$. 
Now consider the terms involving $P(t)$ in the objective of \textbf{P6}: $[U(t) + Z(t) - V C(P(t))]P(t)$.
Since $U(t) > V \chi_{min}$, using property (\ref{eq:chi_defn}), we have:
\begin{align*}
[U(t) + Z(t) - &V C(P(t))]P(t) \leq \\ 
&[U(t) + Z(t) - V C(P_{peak})]P_{peak}
\end{align*}
Thus, the optimal solution to problem \textbf{P6} chooses
$P^*(t) = P_{peak}$.
Now, let $R^*(t), D^*(t)$ and $\gamma^*(t)$ denote the other control decisions by the optimal solution to \textbf{P6}. 
Then the amount of power remaining for the data center (after recharging or discharging the battery) is 
$P_{peak} - R^*(t) + D^*(t)$. Out of this, a fraction $1 - \gamma^*(t)$ is used to serve the delay intolerant workload. Thus:
\begin{align*}
(1 - \gamma^*(t))[P_{peak} - R^*(t) + D^*(t)] = W_2(t) 
\end{align*}
Using this, and the fact that $R^*(t) \leq R_{max}$, we have:
\begin{align*}
&\gamma^*(t)[P_{peak} - R^*(t) + D^*(t)] \\
&= P_{peak} - R^*(t) + D^*(t) - W_2(t) \\
&\geq P_{peak} - R_{max} - W_2(t) \geq W_1(t)
\end{align*}
where we used the fact that $W_1(t) + W_2(t) \leq W_{max} \leq P_{peak} - R_{max}$.
Thus, using (\ref{eq:u_t}), it can be seen that the amount of new arrivals to $U(t)$ cannot exceed the
total service and this yields $U(t+1) \leq U(t)$.

(\ref{eq:z_max}) can be shown by similar arguments.
Clearly, (\ref{eq:z_max}) holds for $t=0$.
Now suppose it holds for slot $t$.  We will show that it also holds for slot $t+1$.
First suppose $Z(t) \leq V \chi_{min}$. Then, by (\ref{eq:z_t}),
the most that $Z(t)$ can increase in one slot is $\epsilon$ so that $Z(t+1) \leq V \chi_{min} + \epsilon$. Next,
suppose $V \chi_{min} < Z(t) \leq V \chi_{min} + \epsilon$. Then, by a similar argument as before,
the optimal solution to problem \textbf{P6} chooses $P^*(t) = P_{peak}$
Now, let $R^*(t), D^*(t)$ and $\gamma^*(t)$ denote the other control decisions by the optimal solution to \textbf{P6}.
Then the amount of power remaining for the data center is $P_{peak} - R^*(t) + D^*(t)$.
Out of this, a fraction $(1 - \gamma^*(t))$ is used for the delay intolerant workload. Thus:
\begin{align*}
(1 - \gamma^*(t))[P_{peak} - R^*(t) + D^*(t)] = W_2(t)
\end{align*}

Using this, and the fact that $R^*(t) \leq R_{max}$, we have:
\begin{align*}
&\gamma^*(t)[P_{peak} - R^*(t) + D^*(t)]
\geq P_{peak} - R_{max} - W_2(t) \\
&\geq P_{peak} - R_{max} - W_{2,max} \geq W_{max} - W_{2,max} \geq \epsilon
\end{align*}
where we used the fact that $W_2(t) \leq W_{2,max}$ and $W_{max} \leq P_{peak} - R_{max}$.
Thus, using (\ref{eq:z_t}), it can be seen that the amount of new arrivals to $Z(t)$ cannot exceed the
total service and this yields $Z(t+1) \leq Z(t)$.

(\ref{eq:q_max}) can be shown by similar arguments and the proof is omitted for brevity.
\end{proof}

\begin{proof} (Theorem \ref{thm:two} part $2$)
We first show that (\ref{eq:x_bounds_ext}) holds for $t=0$. Using the definition
of $X(t)$ from (\ref{eq:x_t_def_ext}), we have that 
$Y(0) = X(0) + Q_{max} + D_{max} + Y_{min}$. Since $Y_{min} \leq Y(0)$,
we have:
\begin{align*}
&Y_{min} \leq X(0) +  Q_{max} + D_{max} + Y_{min} \\
\implies &- Q_{max} - D_{max} \leq X(0)
\end{align*}
Next, we have that $Y(0) \leq Y_{max}$, so that:
\begin{align*}
&X(0) +  Q_{max} + D_{max} + Y_{min} \leq Y_{max}\\
\implies & X(0) \leq Y_{max} - Y_{min} - Q_{max} -  D_{max}
\end{align*}
Combining these two shows that $-Q_{max} - D_{max} \leq X(0) \leq Y_{max} - Y_{min} - Q_{max} - D_{max}$. 

Now suppose (\ref{eq:x_bounds_ext}) holds for slot $t$.  We will show that it also holds for slot $t+1$. First, suppose $-V C_{min} < X(t) \leq 
Y_{max} - Y_{min} - Q_{max} - D_{max}$. Then,
from Lemma \ref{lem:five}, we have that $R^*(t) = 0$. Thus, using (\ref{eq:x_t_ext}) we have that $X(t+1) \leq X(t) \leq Y_{max} - Y_{min} - Q_{max} - D_{max}$. 
Next, suppose $X(t) \leq -V C_{min}$. Then, the maximum possible increase is $R_{max}$ so that $X(t+1) \leq -V C_{min} + R_{max}$.
Now for all $V$ such that $0 \leq V \leq V^{ext}_{max}$, we have that
$-V C_{min} + R_{max} \leq Y_{max} - Y_{min} - (D_{max} + W_{1,max} + \epsilon) - V \chi_{min} = Y_{max} - Y_{min} - Q_{max} - D_{max}$.
This follows from the definition (\ref{eq:v_max_ext}) and the fact that $\chi_{min} \geq C_{min}$. 
Using this,
we have $X(t+1) \leq Y_{max} - Y_{min} - Q_{max} - D_{max}$ for this case as well. This
establishes that $X(t+1) \leq Y_{max} - Y_{min} - Q_{max} - D_{max}$. 

Next, suppose $-Q_{max} - D_{max} \leq X(t) < - Q_{max}$. Then, from Lemma \ref{lem:five}, we have that $D^*(t) = 0$. 
Thus, using (\ref{eq:x_t_ext}) we have that 
$X(t+1) \geq X(t) \geq - Q_{max} - D_{max}$.
 Next, suppose $- Q_{max} \leq X(t)$. Then, the
maximum possible decrease is $D_{max}$ 
so that $X(t+1) \geq - Q_{max} - D_{max}$ for this case as well. This shows that
$X(t+1) \geq - Q_{max} - D_{max}$. Combining these two bounds proves (\ref{eq:x_bounds_ext}).
\end{proof}


\begin{proof} (Theorem \ref{thm:two} parts $3$ and $4$) Part $3$ directly follows from (\ref{eq:x_bounds_ext}) and (\ref{eq:x_t_def_ext}). 
Using $Y(t) = X(t) + Q_{max} + D_{max} + Y_{min}$ in the lower bound in (\ref{eq:x_bounds_ext}), we have:
\begin{align*}
&- Q_{max} - D_{max} \leq Y(t) - Q_{max} - D_{max} - Y_{min} \\
&\implies  Y_{min} \leq Y(t)
\end{align*}
Similarly, using $Y(t) = X(t) + Q_{max} + D_{max} + Y_{min}$ in the upper bound in (\ref{eq:x_bounds_ext}), we have:
\begin{align*}
& Y(t) - Q_{max} - D_{max} - Y_{min} \leq Y_{max} - Y_{min} - Q_{max} - D_{max}\\
&\implies Y(t) \leq Y_{max}
\end{align*}
Part $4$ now follows from part $3$ and the constraint on $P(t)$ in \textbf{P6}. 
\end{proof}

\begin{proof} (Theorem \ref{thm:two} part $5$) This follows from part $1$ and Lemma \ref{lem:worst_case_delay}.
\end{proof}

\begin{proof} (Theorem \ref{thm:two} part $6$) 
We use the following Lyapunov function: $L(\boldsymbol{Q}(t)) \defequiv \frac{1}{2}(U^2(t) + Z^2(t) + X^2(t))$. 
Define the conditional $1$-slot Lyapunov drift as follows:
\begin{align}
\Delta(\boldsymbol{Q}(t)) \defequiv \expect{L(\boldsymbol{Q}(t+1)) - L(\boldsymbol{Q}(t))|\boldsymbol{Q}(t)}
\label{eq:drift1_ext}
\end{align}
Using (\ref{eq:u_t}), (\ref{eq:z_t}), (\ref{eq:x_t_ext}), the drift + penalty term 
can be bounded as follows:
\begin{align}
&\Delta(\boldsymbol{Q}(t)) + V \expect{P(t)C(t) + 1_R(t)C_{rc} + 1_D(t)C_{dc}|\boldsymbol{Q}(t)} \leq \nonumber \\ 
&B_{ext} - [U(t)+Z(t)]\expect{P(t)|\boldsymbol{Q}(t)} - W_2(t)[U(t) + Z(t)]\nonumber\\
&- [X(t)+U(t)+Z(t)]\expect{D(t) - R(t)|\boldsymbol{Q}(t)} \nonumber \\
&+ V \expect{P(t)C(t) + 1_R(t)C_{rc} + 1_D(t)C_{dc}|\boldsymbol{Q}(t)}
\label{eq:drift2_ext}
\end{align}
where $B_{ext} = (P_{peak} + D_{max})^2 + \frac{(W_{1,max})^2 + \epsilon^2}{2} + B$. 
Comparing this with \textbf{P6}, it can be seen that given any queue value $X(t)$,
our control algorithm is designed to \emph{minimize} the
right hand side of (\ref{eq:drift2_ext}) over all possible feasible control policies. This includes the
optimal, stationary, randomized policy given in Lemma $\ref{lem:three}$. Using the same argument as before,
we have the following:
\begin{align*}
&\Delta(\boldsymbol{Q}(t))+ V \expect{P(t)C(t) + 1_R(t)C_{rc} + 1_D(t)C_{dc}|\boldsymbol{Q}(t)}  \leq \nonumber \\
& B_{ext} + V \expect{\hat{P}(t)\hat{C}(t) + \hat{1}_R(t)C_{rc} + \hat{1}_D(t)C_{dc}|\boldsymbol{Q}(t)} \nonumber \\
& = B_{ext} + V \hat{\phi}_{ext}
\end{align*}
Taking the expectation of both sides and using the law of iterated expectations 
and summing over $t \in \{0, 1, 2, \ldots, T-1\}$, we get
\begin{align*}
&\sum_{t=0}^{T-1} V \expect{P(t)C(t) + 1_R(t)C_{rc} + 1_D(t)C_{dc}} \leq \\
& B_{ext}T + VT \hat{\phi}_{ext} - \expect{L(\boldsymbol{Q}(T))} + \expect{L(\boldsymbol{Q}(0))}
\end{align*}
Dividing both sides by $VT$ and taking limit as $T \rightarrow\infty$ yields:
\begin{align*}
\lim_{T\rightarrow\infty} \frac{1}{T} \sum_{t = 0}^{T-1}& \expect{P(t) C(t) + 1_R(t) C_{rc} + 1_D(t) C_{dc}} \leq \\
\hat{\phi}_{ext} + B_{ext}/V
\end{align*}
where we used the fact that $\expect{L(\boldsymbol{Q}(0))}$ is finite and 
that $\expect{L(\boldsymbol{Q}(T))}$ is non-negative.
\end{proof}

\section*{Appendix E - Full Solution for Convex, Increasing $C(t)$}

Let $C'(S, P)$ be the derivative of  $\hat{C}(S, P)$ with respect to $P$. Also, let $P'$ denote the solution to 
the equation $C'(S, P) = 0$ and $C(P') = \hat{C}(S, P')$. 
Then, the optimal solution can be obtained as follows:
\begin{enumerate}
\item If $P_{low} \leq P' \leq W(t)$, then $R^*(t)=0$ and we have the following two cases:
\begin{enumerate}
\item If $P'(X(t) + V C(P')) + V C_{dc} < \theta(t)$, then $P^*(t) = P'$, $D^*(t) = W(t) - P'$.
\item Else, draw all power from the grid. This yields $D^*(t) = 0$ and $P^*(t) = W(t)$.
\end{enumerate}

\item If $W(t) < P' \leq P_{high}$, then $D^*(t)=0$ and we have the following two cases:
\begin{enumerate}
\item If $P'(X(t) + V C(P')) + V C_{rc} < \theta(t)$, then $P^*(t) = P'$, $R^*(t) = P' - W(t)$.
\item Else, draw all power from the grid. This yields $R^*(t) = 0$ and $P^*(t) = W(t)$.
\end{enumerate}

\item If $P_{high} < P'$, then $R^*(t)=0$ and we have the following two cases:
\begin{enumerate}
\item If $P_{high}(X(t) + V C(P_{high})) + V C_{dc} < \theta(t)$, then $P^*(t) = P_{high}$, $D^*(t) = W(t) - P_{high}$.
\item Else, draw all power from the grid. This yields $D^*(t) = 0$ and $P^*(t) = W(t)$.
\end{enumerate}

\item If $P' < P_{low}$, then $D^*(t)=0$ and we have the following two cases:
\begin{enumerate}
\item If $P_{low}(X(t) + V C(P_{low})) + V C_{rc} < \theta(t)$, then $P^*(t) = P_{low}$, $R^*(t) = P_{low} - W(t)$.
\item Else, draw all power from the grid. This yields $R^*(t) = 0$ and $P^*(t) = W(t)$.
\end{enumerate}
\end{enumerate}


\end{document}